\documentclass[twocolumn,showpacs,preprintnumbers,final,aps,citeautoscript,nofootinbib,prb,superscriptaddress,%dvipdfmx
]{revtex4-2}
\usepackage{amsmath,amssymb,amsfonts}
\usepackage{physics}
\usepackage{graphicx}% Include figure files
\graphicspath{{./figures/}}
\usepackage{float}
\usepackage{bm}% bold math
\usepackage{times}
\usepackage[colorlinks,bookmarks=true,citecolor=blue,linkcolor=blue,urlcolor=blue,breaklinks=true]{hyperref}
\usepackage{amsthm}
\usepackage{mathrsfs}
\usepackage{mleftright}
\usepackage{mathtools}
\newtheorem{thm}{Theorem}
\newtheorem{lem}{Lemma}

\newcommand{\rmm}[1]{{\rm{#1}}}
\newcommand{\bmm}[1]{{\bm{#1}}}

\newcommand{\del}[2]{ \frac{\partial #1}{\partial #2}}

\newcommand{\eno}[1]{\rmm{e}^{#1}}

\newcommand{\lr}[3] { \left#1 #2 \right#3}
\newcommand{\hc}[1]{ {#1} ^{\dagger} }

\newcommand{\com}[2]{\lr{[}{#1,#2}{]} }
\newcommand{\kitaiti}[1]{ \lr{\langle} {#1} {\rangle} }

\newcommand{\chuukakko}[1]{\lr{\{}{#1}{\}}}

\newcommand{\acom}[2]{\chuukakko{#1,#2}}

\newcommand{\migi}{\rightarrow}
\newcommand{\hdr}{\leftarrow}

\newcommand{\hf}{\frac{1}{2}}

\newcommand{\eps}{\epsilon}

\newcommand{\up}{\uparrow}
\newcommand{\dw}{\downarrow}
\newcommand{\sx}{\sigma^{x}}
\newcommand{\sy}{\sigma^{y}}
\newcommand{\sz}{\sigma^{z}}

\newcommand{\im}{\rmm{i}}
\newcommand{\nt}{\notag \\}

\newcommand{\FP}{\qty(-1)^F}
\begin{document}
%%%%%%%%%%%%%%%%%%%%%%%%%%%%%%%%%%%%%%%%%%%%%%%%%%%%%%%%%%%%%%%
%%%%%%%%%%%%%%%%%%%%%%%%%%%%%%%%%%%%%%%%%%%%%%%%%%%%%%%%%%%%%%%
\title{Interacting Kitaev chain with $\mathcal{N}=1$ supersymmetry}
\preprint{Phys. Rev. B \textbf{109}, 085141 (2024); DOI: 10.1103/PhysRevB.109.085141}
%%%%%%%%%%%%%%%%%%%%%%%%%%%%%%%%%%%%%%%%%%%%%%%%%%%%%%%%%%%%%%%
\author{Urei Miura} 
\email{urei.miura@yukawa.kyoto-u.ac.jp}
\affiliation{%
Division of Physics and Astronomy, Graduate School of Science, 
Kyoto University, Kyoto 606-8502, Japan}%
\affiliation{Center for Gravitational Physics and Quantum Information, Yukawa Institute for Theoretical Physics, Kyoto University, Kitashirakawa Oiwake-Cho, Kyoto 606-8502, Japan}%
%%%%%%%%%%%%%%%
\author{Kenji Shimomura} 
\affiliation{%
Division of Physics and Astronomy, Graduate School of Science, 
Kyoto University, Kyoto 606-8502, Japan}%
\affiliation{Center for Gravitational Physics and Quantum Information, Yukawa Institute for Theoretical Physics, Kyoto University, Kitashirakawa Oiwake-Cho, Kyoto 606-8502, Japan}%
%%%%%%%%%%%%%%%
\author{Keisuke Totsuka}
\affiliation{Center for Gravitational Physics and Quantum Information, Yukawa Institute for Theoretical Physics, Kyoto University, Kitashirakawa Oiwake-Cho, Kyoto 606-8502, Japan}%
%%%%%%%%%%%%%%%
%%%%%%%%%%%%%%%%%%%%%%%%%%%%%%%%%%%%%%%%%%%%%%%%%%%%%%%%%%
\date{Received 6 November 2023; revised 5 February 2024; accepted 7 February 2024; published 26 February 2024}
%%%%%%%%%%%%%%%%%%%%%%%%%%%%%%%%%%%%%%%%%%%%%%%%%%%%%%%%%%
\begin{abstract}
Lattice models with supersymmetry are known to exhibit a variety of remarkable properties that do not exist in the relativistic models. In this paper, we introduce an interacting generalization of the Kitaev chain of Majorana fermions with $\mathcal{N} = 1$ supersymmetry and investigate its low-energy properties, paying particular attention to the ground-state degeneracy and low-lying fermionic excitations. First, we establish the existence of a phase with spontaneously broken supersymmetry and a phase transition out of it with the help of variational arguments and the exact ground state. We then develop, based on the superfield formalism, a simple mean-field theory, in which the order parameters detect supersymmetry breaking, to understand the ground-state phases and low-lying Nambu-Goldstone fermions. At the solvable point ({\em frustration-free point}), the exact ground state of an open chain exhibits large degeneracy of the order of the system size, which is attributed to the existence of a zero-energy domain wall (dubbed kink or skink) separating the topological and trivial states of Majorana fermions. Our results may shed new light on the intriguing ground-state properties of supersymmetric lattice models.

\end{abstract}
%%%%%%%%%%%%%%%%%%%%%%%%%%%%%%%%%%%%%%%%%%%%%%%%%%%%%%%%%
\maketitle
%%%%%%%%%%%%%%%%%%%%%%%%%%%%%%%%%%%%%%%%%%%%%%%%%%%%%%%%%
%\tableofcontents

%%%%%%%%%%%%%%%%%%%%%%%%%%%%%%%%%%%%%%%%%%%%%%%%%%%%%%%%%
\section{INTRODUCTION}
%%%%%%%%%%%%%%%%%%%%%%%%%%%%%%%%%%%%%%%%%%%%%%%%%%%%%%%%%
Symmetry has been playing a vital role in modern physics.
For instance, it provides unifying viewpoints to the list of existing elementary particles and is also indispensable in labeling and classifying different phases of matter.  In contrast to the usual bosonic symmetry, supersymmetry which is generated by fermionic operators has many remarkable properties. 
These properties are believed to contribute to the resolution of the hierarchy problem and the unification of forces in particle physics \cite{PhysRevD.13.974,PhysRevD.14.1667}. 
In condensed matter or statistical physics, on the other hand, supersymmetry has been discussed in different contexts. For instance, it has been utilized for a long time as a convenient theoretical tool, e.g., in the study of disordered systems \cite{efetov1999supersymmetry,cecotti1983stochastic,gozzi1993stochastic,parisi1979random}, whereas there have also been various proposals for the physical realizations of supersymmetry.  
Probably, the most classic example is the emergent $\mathcal{N}=1$ supersymmetry ($\mathcal{N}$ being the number of fermionic generators) at the tricritical point of the two-dimensional Ising model with vacancies \cite{Friedan-Q-S-SUSY-85,Qiu-86}. 
Other arenas for supersymmetry include topological insulators and superconductors \cite{grover2014emergent,qi2009time}.

While in many nonrelativistic applications mentioned earlier, supersymmetry concerns actual bosons and fermions (as in Bose-Fermi mixtures in cold gases \cite{yu2008supersymmetry,yu2010simulating,snoek2005ultracold,snoek2006theory,lozano20071+,shi2010supersymmetric,lai2015relaxation,blaizot2015spectral,blaizot2017goldstino,tajima2021goldstino}), there has recently been a growing interest in lattice models of supersymmetry that exclusively involve fermions.  
Since the pioneering work by Nicolai \cite{nicolai1976supersymmetry}, a variety of properties that do not exist in the relativistic counterparts have been uncovered \cite{nicolai1976supersymmetry,nicolai1977extensions,fendley2003lattice,fendley2005exact,huijse2008superfrustration,huijse2012supersymmetric,sannomiya2016supersymmetry,sannomiya2017supersymmetry,sannomiya2019supersymmetry}. Among the main issues in this field are significant ground-state degeneracy dubbed superfrustration \cite{huijse2008superfrustration,huijse2012supersymmetric} and the spontaneous breaking of supersymmetry. 

Yet another important problem is the nature of the gapless fermionic excitations [the Nambu-Goldstone (NG) fermions] that are expected to appear when supersymmetry is spontaneously broken \cite{Salam-S-74,witten1982constraints}. However, unlike in the case of NG bosons (including those in nonrelativistic systems) for which comprehensive classification theories are known \cite{PhysRevLett.108.251602,PhysRevLett.110.091601}, little is known for their fermionic counterparts except in the relativistic cases. In fact, some supersymmetric lattice models are known to exhibit exotic dispersion relations \cite{sannomiya2017supersymmetry,sannomiya2019supersymmetry,fendley2019free}, requiring further research in this direction.

In realizing supersymmetry in purely fermionic systems, of particular interest are Majorana fermions which are anticipated to appear in topological superconductors and hold potential applications in quantum computing (see, e.g., Refs.~\cite{Nayak-S-S-F-S-08,Alicea-review-12} for reviews).
For instance, suggestions for experimentally realizable supersymmetry with Majorana systems are detailed in, e.g.,  Ref.~\cite{huang2017supersymmetry}. Also, one-dimensional Majorana systems that host the above-mentioned supersymmetric tricriticality with the central charge $c=\frac{7}{10}$ have been proposed in Refs.~ \cite{rahmani2015emergent,li2018numerical,o2018lattice}. An even larger $\mathcal{N}=4$ extended supersymmetry in Majorana nanowires is also discussed \cite{marra20221d}. 

Despite its exotic look, supersymmetry can appear quite generally in condensed-matter systems either in emergent or explicit manners. In fact, even a general recipe to realize supersymmetry in systems of Majorana fermions has been proposed \cite{hsieh2016all}. However, many previous studies either explored supersymmetry emergent at low energies or took an approach where the locality of the supercharges was lost by, e.g., taking the square root of the Hamiltonian. Therefore, constructing supersymmetric models of interacting lattice fermions that possess {\em local} supercharges and display intriguing properties would be an interesting direction to pursue. 

Inspired by the results in Ref.~\cite{sannomiya2019supersymmetry}, we introduce in this paper a model of lattice Majorana fermions with explicit $\mathcal{N}=1$ supersymmetry, whose supercharge is given by a sum of local terms, and investigate its ground-state properties and low-lying fermionic excitations with both analytical and numerical methods. 
One advantage of studying systems with supersymmetry is that the zero-energy ground state, if it exists, can be easily identified as one annihilated by the supercharge. We use this property to find all the exact ground states in a specific case, which helps us to show that the model displays a supersymmetry-restoring transition. Interestingly, this exact ground state exhibits large degeneracy (proportional to the system size) due to a zero-energy domain wall that separates the topological and trivial states of Majorana fermions.   

The organization of this paper is as follows.  
In Sec.~\ref{sec:Model}, we introduce a model of interacting Majorana fermions on a one-dimensional lattice. It possesses a fermionic conserved charge (supercharge) that guarantees that the model is manifestly supersymmetric already at the lattice level. Some discussions on its symmetries and the equivalent spin model are also presented. 
In Sec.~\ref{sec:SUPERSYMMETRY BREAKING}, on top of the standard one, we introduce a modified definition of spontaneous breaking of supersymmetry, which turns out to be useful in obtaining consistent results for different boundary conditions. Using a variational inequality and the exact ground states at a particular point, we identify the parameter range within which supersymmetry is broken (according to the modified definition). 

We develop a mean-field theory based on the superfield formalism in Sec.~\ref{sec:maenfield} to investigate the breaking of supersymmetry and the emergence of NG fermions. This approach not only predicts that the NG fermions exhibit linear dispersion but also yields results consistent with those of other approaches.
In Sec.~\ref{skinksection}, we first present numerical results on the spectrum in which we observe that, under the open boundary condition, the system at a special value of the parameter exhibits ground-state degeneracy which is proportional to the system size. We then explain it by a superdoublet of domain walls (dubbed kink and skink) connecting the topological and trivial phases of the Majorana fermion. This degeneracy is resolved if we move away from the special point, and using the first-order perturbation we show that these domain-wall states form a gapless spectrum (behaving like $\sim k^{2}$) over the doubly degenerate ground state. 
We conclude the paper in Sec.~\ref{sec:conclusion}.
%%%%%%%%%%%%%%%%%%%%%%%%%%%%%%%%%%%
%%%%%%%%%%%%%%%%%%%%%%%%%%%%%%%%%%%%%%%%%%%%%%%%%%%%%%%%%%
\section{MODEL}
\label{sec:Model}
%%%%%%%%%%%%%%%%%%%%%%%%%%%%%%%%%%%%%%%%%%%%%%%%%%%%%%%%%%
%%%%%%%%%%%%%%%%%%%%%%%%%%%%%%%%%%%%%%%%%%%%%%%%%%%%%%%%%%
In this section, we introduce a supersymmetric lattice Majorana fermion model 
which has an explicit $\mathcal{N}=1$ supersymmetry and describe its symmetries. 

%%%%%%%%%%%%%%%%%%%%%%%%%%%%%%%%%%%%%%%%%%%%%%%%%%%%%%%%%%%
\subsection{Supercharge and Hamiltonian}
%%%%%%%%%%%%%%%%%%%%%%%%%%%%%%%%%%%%%%%%%%%%%%%%%%%%%%%%%%%
The $\mathcal{N}=1$ supersymmetry (SUSY) is a fermionic symmetry defined by the superalgebra 
of Hermitian operators $R$ (supercharge), $\qty(-1)^F$ (fermion parity), and $H$ (Hamiltonian) as given in the following equations: 
\cite{witten1981dynamical,witten1982constraints,cooper1995supersymmetry}
\begin{equation}
\label{superalgebra}
\begin{split}
    &H=\hf R^2=\frac{1}{4}\acom{R}{R} \; , \\
    &\acom{R}{\qty(-1)^F}=0 \; , \\
    &\qty{\qty(-1)^F}^2=+1  \; .
\end{split}
\end{equation}
 and all the non-trivial spectral properties follow immediately from 
the following fundamental relations:
\begin{equation}
\begin{split}
& \com{H}{R} =0 \; , \quad
\com{H}{(-1)^{F}} =0  \; .
\end{split} 
\label{eqn:fundamental-rel}
\end{equation}
Our approach first defines the supercharge 
$R$ such that it anticommutes with the fermion parity 
$\FP$, which is defined from Majorana fermions. Then, we determine the Hamiltonian $H$ as the square of $R$.

In this paper, we consider a one-dimensional lattice with $L$ sites.  
On the $j$th site, there exist two Majorana fermions, $\beta_j=\hc{\beta}_j$ and 
$\gamma_j=\hc{\gamma}_j$ (Fig.~\ref{fig:01})
which obey the standard Clifford algebra:
\begin{align}
& \acom{\beta_{i}}{\beta_{j}}=\acom{\gamma_{i}}{\gamma_{j}}=2\delta_{i,j} \nt
&\acom{\beta_{i}}{\gamma_{j}}=0  
\end{align}
for all $i,j=1,\ldots,L$. 
%%%%%%%%%%%%%FIG%%%%%%%%%%%%%%%%%%
\begin{figure}[htbp]
  \includegraphics[width=\columnwidth,clip]{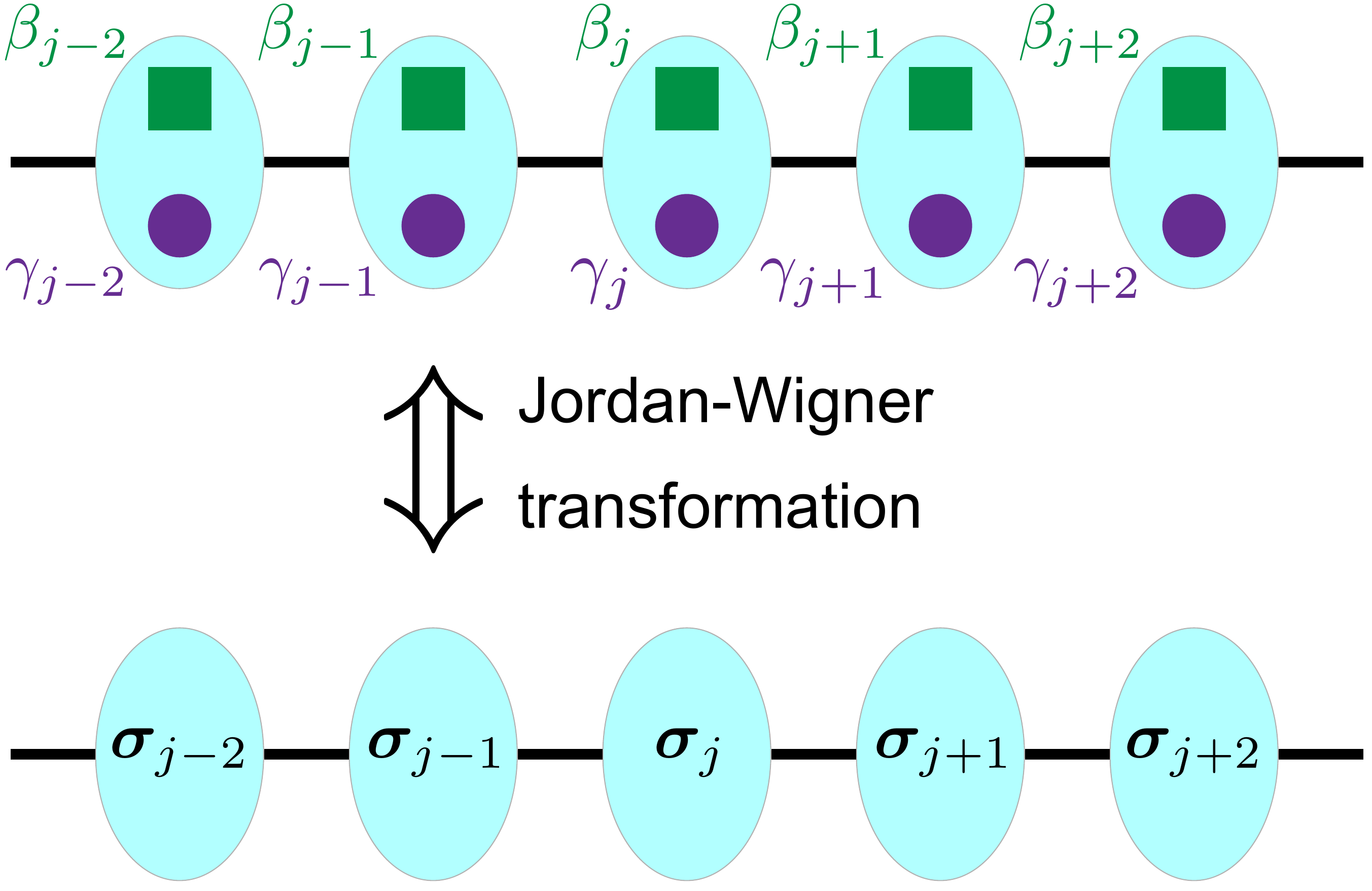}
  \caption{ 
One-dimensional lattice on which the model \eqref{eq:model} is defined. 
The squares and circles on the individudal sites are (represented by large blue circles) are the local Majorana fermions $\beta_j$ and $\gamma_j$, 
respectively.  The two-dimensional fermion Fock space on each site spanned by $\beta_{j}$ and $\gamma_{j}$ is transformed to 
a two-dimensional Hilbert space of a spin 1/2 ($\bmm{\sigma}_{j}$) via the Jordan-Wigner transformation and vice versa. 
\label{fig:01}}
\end{figure}
%%%%%%%%%%%%%FIG%%%%%%%%%%%%%%%%%%
With these Majorana fermions, we define a one-parameter family of the real supercharge $R(\alpha)$ in the periodic boundary condition (PBC) as:  
\begin{equation}
\label{eqn:superchargePBC}
  R^{\text{PBC}}\qty(\alpha)=\sum_{j=1}^{L} \qty{\qty(1-\alpha)\beta_j+\alpha \im \beta_j\beta_{j+1}\gamma_j}
      \quad (0\leq \alpha\leq1) \; ,
\end{equation}
where $\beta_{L+j}:=\beta_{j}$ and $\gamma_{L+j}:=\gamma_{j}$, and in the open boundary condition (OBC), as:
\begin{equation}
\label{eqn:superchargeOBC}
  R^{\text{OBC}}\qty(\alpha) =\qty(1-\alpha)\sum_{j=2}^{L-1} \beta_j+\alpha \sum_{j=1}^{L-1}  \im \beta_j\beta_{j+1}\gamma_j
      \quad (0\leq \alpha\leq1) \; .
\end{equation}
By direct calculation, we see that the Hamiltonian $H\qty(\alpha)$ takes the form 
\begin{equation}
  \label{eq:model}
\begin{split}
 H\qty(\alpha) & = \frac{1}{2} R(\alpha)^{2}  \\
 & =  \qty(1-\alpha)^2 H_0 + \alpha\qty(1-\alpha) H_\text{CTI} + \alpha^2 H_\text{int} \; .
 \end{split}
\end{equation}
The expressions of the three terms $H_0$, $H_\text{CTI}$, and $H_\text{int}$ depend on the boundary conditions.  
Specifically, they are given by:
\begin{subequations}
\begin{align}
  &H^{\text{PBC}}_0 =\frac{L}{2}
  \label{eqn:H_0_PBC} \\
  &H^{\text{PBC}}_\text{CTI} =\sum_{j=1}^{L} \qty(\im \beta_{j+1}\gamma_j-\im\beta_{j-1}\gamma_{j-1})  
  \label{eqn:H_CTI_PBC} \\
\begin{split}
  &H^{\text{PBC}}_\text{int}
  =\hf \qty(\sum_{j=1}^{L}\im \beta_j\beta_{j+1}\gamma_j)^2  
  =\frac{L}{2}-\sum_{j=1}^{L}\beta_{j-1}\beta_{j+1} \gamma_{j-1} \gamma_{j}  
\end{split}
\label{eqn:H_int_PBC}
\end{align}
\end{subequations}
for PBC, and 
%%%%%%%%%%%%%PBC->OBC%%%%%%%%%%%%%%%%%%%%%%%%%%%%%%%%%%%%%%%%%%%%%%%%%%%%
\begin{subequations}
\begin{align}
  &H^{\text{OBC}}_0 =\frac{L-2}{2}
  \label{eqn:H_0_OBC} \\
  &H^{\text{OBC}}_\text{CTI} =\sum_{j=2}^{L-1} \qty(\im \beta_{j+1}\gamma_j-\im\beta_{j-1}\gamma_{j-1})  
  \label{eqn:H_CTI_OBC} \\
\begin{split}
 & H^{\text{OBC}}_\text{int}
  =\hf \qty(\sum_{j=1}^{L-1}\im \beta_j\beta_{j+1}\gamma_j)^2  \\
 &  \phantom{ H_\text{int}^{\text{OBC}} }
 =\frac{L-1}{2}-\sum_{j=2}^{L-1}\beta_{j-1}\beta_{j+1} \gamma_{j-1} \gamma_{j} 
\end{split}
\label{eqn:H_int_OBC}
\end{align}
\label{eqn:H^OBC}
\end{subequations}
for OBC.  
The $O\qty(\alpha\qty(1-\alpha))$ part $H_\text{CTI}$ is nothing but the Kitaev chain \cite{kitaev2001unpaired} of noninteracting Majorana fermions at its critical point, and 
the $O(\alpha^{2})$ term $H_{\text{int}}$ introduces interactions among them.  
The operator $(-1)^{F}$ appearing in the supersymmetric algebra \eqref{superalgebra} is given by the physical fermion parity:
\begin{align}
  \qty(-1)^F\coloneqq\prod_{j=1}^L \im \beta_j \gamma_j \; ,
\label{eqn:F-parity-Majorana}
\end{align}
which, by the relations \eqref{superalgebra}, is conserved.

%%%%%%%%%%%%%%%%%%%%%%%%%%%%%%%%%%%%%%%%%%%%%%%%%%%%%%%%%
\subsection{Equivalent spin model}
\label{sec:spin-rep}
%%%%%%%%%%%%%%%%%%%%%%%%%%%%%%%%%%%%%%%%%%%%%%%%%%%%%%%%%
For practical purposes, it is convenient to express the fermionic model \eqref{eq:model} as an equivalent spin model.  
To this end, we use the Jordan-Wigner transformation:
\begin{equation}
\begin{split}
&\beta_i=\sx_i\prod_{j\qty(<i)}\qty(-\sz_j)   \, , \quad 
\gamma_i=\sy_i\prod_{j\qty(<i)}\qty(-\sz_j)
\end{split}
\label{eqn:JW-tr}
\end{equation}
which transforms a pair of Majorana fermions $(\beta_i,\gamma_i)$ into a set of Pauli matrices 
$\bmm{\sigma}_{i}$. 
The conserved fermion parity \eqref{eqn:F-parity-Majorana} is now translated, in the spin language, into:
\begin{equation}
\qty(-1)^F=\prod_{j=1}^L \qty(-\sz_j)   \;  .
\label{eqn:F-parity-spin}
\end{equation}
The two fermion states with $\im \beta_{j}\gamma_j= \mp 1$ are translated into the spin eigenstates $\sz = \pm 1$ which will be denoted by: 
$\ket{\up}$ ($\sz=+1$) and $\ket{\dw}$ ($\sz= -1$), respectively. 

The two terms $H_\text{CTI}$ and $H_\text{int}$ of the Hamiltonian $H\qty(\alpha)$ \eqref{eq:model} 
are transformed differently depending on the boundary conditions.  
For OBC, Eqs.~\eqref{eqn:H_CTI_OBC} and \eqref{eqn:H_int_OBC} are now replaced by:
\begin{equation}
\begin{split}
&H_\text{CTI}^{\text{OBC}} =\sum_{j=2}^{L-1}
\qty(-\sx_j\sx_{j+1}+\sz_{j-1})
   \\
& H_\text{int}^{\text{OBC}} =\frac{L-1}{2}- \sum_{j=2}^{L-1}\sz_{j-1}\sx_{j}\sx_{j+1}   \;  .
\end{split}
\label{eq:spin-rep-OBC}
\end{equation}
Physically, $H_\text{CTI}$ is the transverse-field Ising model at the critical point, and $H_\text{int}$ adds three-spin interactions.  
Special care must be taken in the case of PBC due to the nonlocal nature of the Jordan-Wigner transformation \eqref{eqn:JW-tr}. 
Specifically, the bulk terms in \eqref{eqn:H_CTI_PBC} and \eqref{eqn:H_int_PBC} are transformed into the usual expressions:
\begin{equation}
\begin{split}
&H_\text{CTI,bulk}^{\text{PBC}} =-\sum_{j=1}^{L-1}
\sx_j\sx_{j+1}+\sum_{j=1}^{L}\sz_j
   \\
& H_\text{int,bulk}^{\text{PBC}} =\frac{L}{2}- \sum_{j=2}^{L-1}\sz_{j-1}\sx_{j}\sx_{j+1}   \;  ,
\end{split}
\label{eq:spin-rep-PBC}
\end{equation}
while the boundary terms after the Jordan-Wigner transformation now acquire the parity-dependence:
\begin{equation}
\label{eq:boundary}
\begin{split}
H_{\text{boundary}}^{\text{PBC}}=& \alpha\qty(1-\alpha)\qty(-1)^F\sx_L\sx_1+\alpha^2\qty(-1)^F\sz_{L-1}\sx_L\sx_1  \\
& - \alpha^2\sz_{L}\sx_1\sx_2  \; ,
\end{split}
\end{equation}
from which we can read off the following fermion-parity-dependent boundary condition for the resulting spin system:
\begin{equation}
\bmm{\sigma}_{j+L} = - (-1)^{F} \bmm{\sigma}_{j} \; .
\label{eqn:BC-for-spin}
\end{equation}

By construction, the resulting spin 1/2 model possesses hidden $\mathcal{N}=1$ supersymmetry.  
To construct the entire supersymmetric spectrum of the original fermionic model \eqref{eq:model} (with PBC)  
in the language of the spin 1/2 model \eqref{eq:spin-rep-PBC}, we need to ``sew'' together the $(-1)^{F}=+1$ sector of the antiperiodic spin chain and 
the $(-1)^{F}= -1$ sector of the periodic one.   
It is important to note that the spin chain with a {\em fixed} boundary condition (e.g., PBC) does not possess supersymmetry.  
In what follows, numerical simulations for spin systems will explicitly take into account the boundary condition \eqref{eqn:BC-for-spin}.

Last, the model only with the interaction part $H_\text{int}$ is integrable by mapping to free fermions through a nonlocal transformation \cite{fendley2019free} 
and is known to exhibit the following properties when defined on an open chain\footnote{To be precise, the author of Ref.~\cite{fendley2019free} considers a slightly different model [see Eq.~(1.6) in the paper] which is given by a sum of $H_{\text{int}}$ and its Kramers-Wannier dual. However, the basic properties (solvability, extensive degeneracy, etc.) described here are carried over to $H_{\text{int}}$ as well.}.  
According to the exact solution, it is gapless and does not break supersymmetry in the infinite-size limit.  
Also, it does not possess the Lorentz invariance even in the low-energy limit and, instead, exhibits extensive degeneracy at each energy level, 
which is often referred to as ``superfrustration'' \cite{huijse2008superfrustration,huijse2012supersymmetric,fendley2005exact,PhysRevLett.101.146406}.  

%%%%%%%%%%%%%%%%%%%%%%%%%%%%%%%%%%%%%%%%%%%%%%%%%%%%%%%%%%%
\subsection{Symmetries}
\label{subSec:symmetry}
%%%%%%%%%%%%%%%%%%%%%%%%%%%%%%%%%%%%%%%%%%%%%%%%%%%%%%%%%%%
The Hamiltonian $H\qty(\alpha)$ defined on a system of size $L$ has the following symmetries.   
Due to the construction \eqref{superalgebra} of the Hamiltonian by the supercharge $R(\alpha)$, the system possesses supersymmetry:
\begin{align}
\com{H\qty(\alpha)}{R\qty(\alpha)}=0 
\end{align}
and conserves the fermion parity:
\begin{align}
\com{H\qty(\alpha)}{\qty(-1)^F}=0 \; .
\end{align}

On the periodic boundary condition, the model (\ref{eq:model}) has the translation symmetry $\mathcal{T}$ that simultaneously shifts $\beta_i$ and $\gamma_i$ by one site:  
$\beta_i \overset{\mathcal{T}}{\mapsto} \beta_{i+1}$, $\gamma_i \overset{\mathcal{T}}{\mapsto} \gamma_{i+1}$.  
Furthermore, the system has an inversion symmetry $\mathcal{I}$ that acts as 
$\beta_i \overset{\mathcal{I}}{\mapsto}   \beta_{L-i \pmod L}$ and 
$\gamma_i \overset{\mathcal{I}}{\mapsto} -\gamma_{L-i-1 \pmod L}$.
The inversion $\mathcal{I}$ preserves the supercharge $\hc{\mathcal{I}} R \qty(\alpha) \mathcal{I} =R \qty(\alpha)$, 
which means that the Hamiltonian $H \qty(\alpha)=\hf R\qty(\alpha)^2$ is kept invariant under $\mathcal{I}$.
Furthermore, this system possesses time-reversal symmetry $\mathscr{T}$, which is defined by the following transformations:
$\im \overset{\mathscr{T}}{\mapsto} -\im$,
$\beta_i \overset{\mathscr{T}}{\mapsto} \beta_{i}$, and
$\gamma_i \overset{\mathscr{T}}{\mapsto} - \gamma_{i}$.

%%%%%%%%%%%%%%%%%%%%%%%%%%%%%%%%%%%%%%%%%%%%%%%%%%%%%%%%%%
%%%%%%%%%%%%%%%%%%%%%%%%%%%%%%%%%%%%%%%%%%%%%%%%%%%%%%%%%%
\section{SUPERSYMMETRY BREAKING}
\label{sec:SUPERSYMMETRY BREAKING}
%%%%%%%%%%%%%%%%%%%%%%%%%%%%%%%%%%%%%%%%%%%%%%%%%%%%%%%%%%
%%%%%%%%%%%%%%%%%%%%%%%%%%%%%%%%%%%%%%%%%%%%%%%%%%%%%%%%%%%
In this section, we prove the existence of spontaneous symmetry breaking (SSB)
of supersymmetry (SUSY) in the supersymmetric lattice model \eqref{eq:model} and estimate the critical value $\alpha_{\text{c}}$ below which 
SUSY is broken. 
As the situation is subtle in the thermodynamic limit, we need to be careful in judging whether or not SUSY is broken in infinite-size systems. 

%%%%%%%%%%%%%%%%%%%%%%%%%%%%%%%%%%%%%%%%%%%%%%%%%%%%%%%
\subsection{Definition of the spontaneous SUSY breaking}
\label{sec:def-SUSY-SSB}
%%%%%%%%%%%%%%%%%%%%%%%%%%%%%%%%%%%%%%%%%%%%%%%%%%%%%%%
As has been emphasized in Ref.~\cite{witten1981dynamical}, the spontaneous breaking of SUSY is rather different from that of ordinary 
global symmetry.  
In fact, there exist models with $\mathcal{N}=2$ where SUSY is always spontaneously broken except at a specific point, even in finite-size systems \cite{sannomiya2016supersymmetry,moriya2018supersymmetry,miura2023supersymmetry}. 
In these models, the ground-state energy is strictly positive, and all eigenstates form SUSY doublets, except at a specific point \cite{witten1982constraints,cooper1995supersymmetry}.

However, the scenario changes when considering infinite-size systems, where the definition of the ground-state energy itself might become uncertain. Notably, demonstrating that $E_{\text{g.s.}}(\alpha;L) > 0$ for any finite $L$ is not sufficient to ascertain SUSY breaking in the limit $L \to \infty$ \cite{witten1982constraints}. In this respect, an alternative definition of SUSY SSB has been proposed \cite{sannomiya2016supersymmetry,sannomiya2017supersymmetry,sannomiya2019supersymmetry}, which focuses not on the ground-state energy itself but on its density.
Specifically, we define SUSY as spontaneously broken in a system of length $L$ (including the thermodynamic limit $L \to \infty$) when the ground-state energy $E_{\text{g.s.}}(\alpha;L)$ per site meets the condition \cite{sannomiya2016supersymmetry,sannomiya2017supersymmetry,sannomiya2019supersymmetry}:
\begin{align}
\frac{E_{\text{g.s.}}(\alpha;L)}{L}
\gneq 0 \; .
\label{eqn:SSB-criterion-gen}
\end{align}
For any finite $L$, this criterion reduces to the usual requirement $E_{\text{g.s.}}(\alpha;L) \gneq 0$ \cite{witten1982constraints}. Moreover, in the infinite-size limit $L \to \infty$, this is equivalent to the strict positivity of the ground-state energy density:
\begin{equation}
e_{\text{g.s.}}(\alpha) \coloneqq
\lim_{L\rightarrow \infty} \frac{E_{\text{g.s.}}(\alpha;L)}{L} \gneq 0 \; .
\label{eqn:SSB-criterion}
\end{equation}
As per the definition \eqref{eqn:SSB-criterion-gen}, for SUSY to be broken even in the infinite-size limit, the positive ground-state energy $E_{\text{g.s.}}(\alpha;L)$ must diverge as $L$ or faster.

Henceforth, we will predominantly consider infinite-size systems, unless otherwise specified, and use Eq.~\eqref{eqn:SSB-criterion} as a criterion for assessing SUSY SSB. Furthermore, as can be seen from Appendix \ref{App:specvary}, the definition of SUSY SSB in the infinite system is independent of the boundary conditions.

Given the above criterion for SUSY SSB, we can immediately conclude that SUSY is broken (in a trivial sense) 
in the model \eqref{eq:model} at least when $\alpha=0$ since $e_{\text{g.s.}}(\alpha=0) = H \qty(\alpha=0)/L=\hf$ there.  
Then, the next question is whether SUSY remains broken even for finite $\alpha$ or not.  
When we increase $\alpha$ from zero, the free Hamiltonian $\alpha\qty(1-\alpha)H_{\rmm{CTI}}$ tends to lower the ground-state energy density as $- 4\alpha\qty(1-\alpha)/\pi$ 
[see Eq.~\eqref{eq:Egs-density}].  
Therefore, if we ignore the $\alpha^{2}$ correction from $\alpha^{2} H_{\rmm{int}}$, $e_{\text{g.s.}}(\alpha)$ stays positive when $\alpha$ is small enough and finally reaches zero at a point, above which SUSY is expected to be restored.  
Other information is available in the $\alpha=1$ where the model reduces to the solvable spin model \cite{fendley2019free} 
mentioned in Sec.~\ref{sec:spin-rep}.  The exact results obtained there rigorously establish unbroken SUSY at the $\alpha=1$.   
The simplest scenario that is compatible with these is that there is a single transition at $\alpha=\alpha_{\text{c}}$ that separates 
the SUSY-broken small-$\alpha$ phase from the symmetric phase at $\alpha\approx1$ (see Fig.~\ref{Fig-SUSY-phasediagram}).
%%%%%%%%%%%%%%%%%%%%%%%%%%%%%%%%%%%%%%%%%%%%%%%%%%%%%%
\subsection{Variational argument for $\alpha_\rmm{c}$}
\label{sec:variational}
%%%%%%%%%%%%%%%%%%%%%%%%%%%%%%%%%%%%%%%%%%%%%%%%%%%%%% 

Now let us give a more rigorous ground to the naive argument presented above.  
In general, when the Hamiltonian is given by a sum of non-commuting components: $H=\sum_i H_i$,  
the true ground state $\ket{0}$ of the full Hamiltonian $H$ contains excited states of the partial Hamiltonians $H_i$.  
Then, by the variational principle, the following inequality (known as Anderson's bound~\cite{anderson1951limits}) holds between the ground-state energy $E_{\text{g.s.}}$ of the entire system 
and those [$E_{\text{g.s.}}^\qty(i)$] of the partial Hamiltonians $H_i$:
\begin{equation}
\label{g.s.ineq}
E_{\text{g.s.}}=\bra{0} H \ket{0}
=\sum_i \bra{0} H_i \ket{0}
\geq \sum_i E_{\text{g.s.}}^\qty(i) 
\end{equation}
(the equality holds when the ground state $\ket{0}$ simultaneously optimizes all $H_i$).

Now we apply the above argument to the Hamiltonian \eqref{eq:model} with $H_{0}=\qty(1-\alpha)^2L/2$, $H_{1}=\alpha\qty(1-\alpha)H_\text{CTI}^{\text{PBC}}$, and 
$H_{2} = \alpha^{2} H_\text{int}^{\text{PBC}}$.   
Using the fact that the ground-state energy density of $H_{\rmm{CTI}}^{\text{PBC}}$ and $H_{\rmm{int}}^{\text{PBC}}$ are given respectively by 
$-\frac{4}{\pi}$ [see Eq.~\eqref{eq:Egs-density}] and 0 \cite{fendley2019free}, we find:
\begin{equation}
e_{\text{g.s.}}(\alpha)  \geq  \frac{\qty(1-\alpha)^2}{2} - \frac{4\alpha\qty(1-\alpha)}{\pi} \; .
\label{eqn:var-lower-bound}
\end{equation}
Thus, we see that SUSY SSB occurs at least in the interval
\begin{equation}
0 \leq  \alpha < \frac{\pi}{8+\pi}=0.281969\ldots  
\label{eqn:SSB-phase-variational}
\end{equation}
in which the right-hand side of \eqref{eqn:var-lower-bound} is positive, e.g., $e_{\text{g.s.}}(\alpha) >0$. 
Therefore, assuming that there is a single phase transition at finite $\alpha$, we can conclude that the critical point $\alpha_{\text{c}}$ 
above which the SUSY-symmetric phase persists satisfies $\frac{\pi}{8+\pi} \leq \alpha_{\text{c}} \leq 1$ (Fig.~\ref{Fig-SUSY-phasediagram}).  
We note that a tighter lower bound $\alpha_\text{c} \geq \frac{\pi}{8}=0.392699\ldots$ has been obtained  
by Sannomiya~\cite{sannomiyaDron}.

%%%%%%%% FIG %%%%%%%%%%%%%%%%%%%%%%%%%%%%%%%%%%%%%%%%%%

\begin{figure}[htbp]
  \includegraphics[width=\columnwidth,clip]{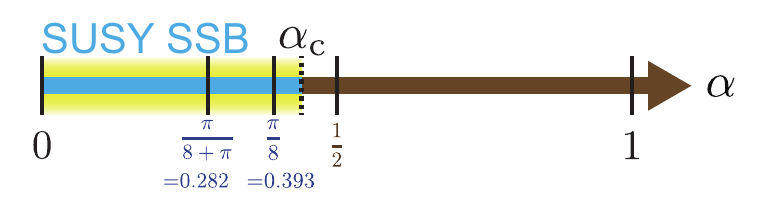}
\caption{%
Preliminary phase diagram of the supersymmetric model \eqref{eq:model} that contains two phases separated by the transition point 
$\alpha_{\text{c}}$.  When $\alpha< \alpha_{\text{c}}$, $e_{\text{g.s.}}(\alpha) >0$ and SUSY is spontaneously broken, while for $\alpha > \alpha_{\text{c}}$, 
SUSY is unbroken.  
The lower bound $\frac{\pi}{8+\pi}$ of $\alpha_{\text{c}}$ which is given by the variational inequality \eqref{eqn:var-lower-bound} is shown 
[a tighter bound $\alpha_{\text{c}} > \pi/8$ obtained in Ref.~\cite{sannomiyaDron} is also shown].  
As is shown in Sec.~\ref{sec:SUSYrestoring}, we can find two {\em exact} zero-energy ground states at $\alpha=\hf$ thereby establishing 
an upper bound $\alpha_{\text{c}} < \hf$.    
  \label{Fig-SUSY-phasediagram}}
\end{figure}
%%%%%%%%%%%%%%%%%%%%%%%%%%%%%%%%%%%%%%%%%%%%%%%%%%%%%%%

%%%%%%%%%%%%%%%%%%%%%%%%%%%%%%%%%%%%%%%%%%%%%%%%%%%
\subsection{Unbroken SUSY at $\alpha=\hf$}
\label{sec:SUSYrestoring}
%%%%%%%%%%%%%%%%%%%%%%%%%%%%%%%%%%%%%%%%%%%%%%%%%%%%
Remarkably, we can find the exact ground states at $\alpha=\hf$, which enable us to unambiguously determine whether SUSY is broken or not.  
We first note that, in the case of periodic boundary conditions ($\beta_{L+j}=\beta_{j}$, $\gamma_{L+j}=\gamma_{j}$), 
the supercharge \eqref{eqn:superchargePBC} can be expressed in two different ways:
\begin{equation}
\begin{split}
R^{\text{PBC}} \qty(\alpha=\hf)&=\hf\sum_{j=1}^L \beta_j\qty(1+\im \beta_{j+1}\gamma_j) \\
&=\hf\sum_{j=0}^{L-1} \beta _{j+1}\qty(1-\im \beta_{j}\gamma_j)  \; .
\end{split}
\label{eqn:R-g1-PBC}
\end{equation}
Now let us consider the states $\ket{\text{A}}$ with $\im \beta_{j+1}\gamma_j= -1$ (for all $i=1,\ldots,L$) 
and $\ket{\text{B}}$ with $\im \beta_{j}\gamma_j=+1$ (for all $i=1,\ldots,L$), which are illustrated in Fig.~\ref{two-GS-PBC}.     
By construction, it is obvious that 
\begin{equation}
R^{\text{PBC}}\qty(\alpha=\hf) \ket{\text{A}} = R^{\text{PBC}}\qty(\alpha=\hf) \ket{\text{B}} = 0 \; ,
\end{equation}
which immediately implies that both are zero-energy ground states of the Hamiltonian $H(\alpha=1/2)$.  
Also, a simple relation: 
\[ (-1)^{F} = \prod_{j=1}^L (\im \beta_j \gamma_j ) = -  \prod_{j=1}^L (- \im \beta_{j+1}\gamma_j) \; , \]
tells that the two states have different eigenvalues of the fermion parity:
\begin{equation}
\qty(-1)^{F} \ket{\text{A}} = - \ket{\text{A}} \; , \quad 
\qty(-1)^{F} \ket{\text{B}} = \ket{\text{B}} .
\end{equation}
Therefore, at least two orthogonal zero-energy ground states $\ket{\text{A}}$ and $\ket{\text{B}}$ exist for any system size $L$ 
thereby implying that SUSY is {\em not} broken at $\alpha=\hf$ 
(in fact, we can show that there is no other zero-energy ground state at $\alpha=\hf$; see Appendix~\ref{sec:proof-2-GS} for the proof).

Note that the two states $\ket{\text{A}}$ and $\ket{\text{B}}$ respectively are regarded as  ``topological'' and ``trivial'' 
in the language of the Kitaev chain \cite{kitaev2001unpaired}.  
In this sense, both topological and trivial states coexist at the special point $\alpha=\hf$.  
Given that $H_\text{CTI}$ is a finely tuned Kitaev chain at the topological-trivial transition, it is plausible to assume that the competition between the two phases at the critical point results in their coexistence.  However, since both states are realized deep inside 
the offcritical regions, the interaction part $H_{\text{int}}$ that is necessary for supersymmetry might play a crucial role 
for the system to stabilize these two particular states.
%%%%%%%%%%%%%FIG%%%%%%%%%%%%%%%%%%
\begin{figure}[htbp]
  \includegraphics[width=\columnwidth,clip]{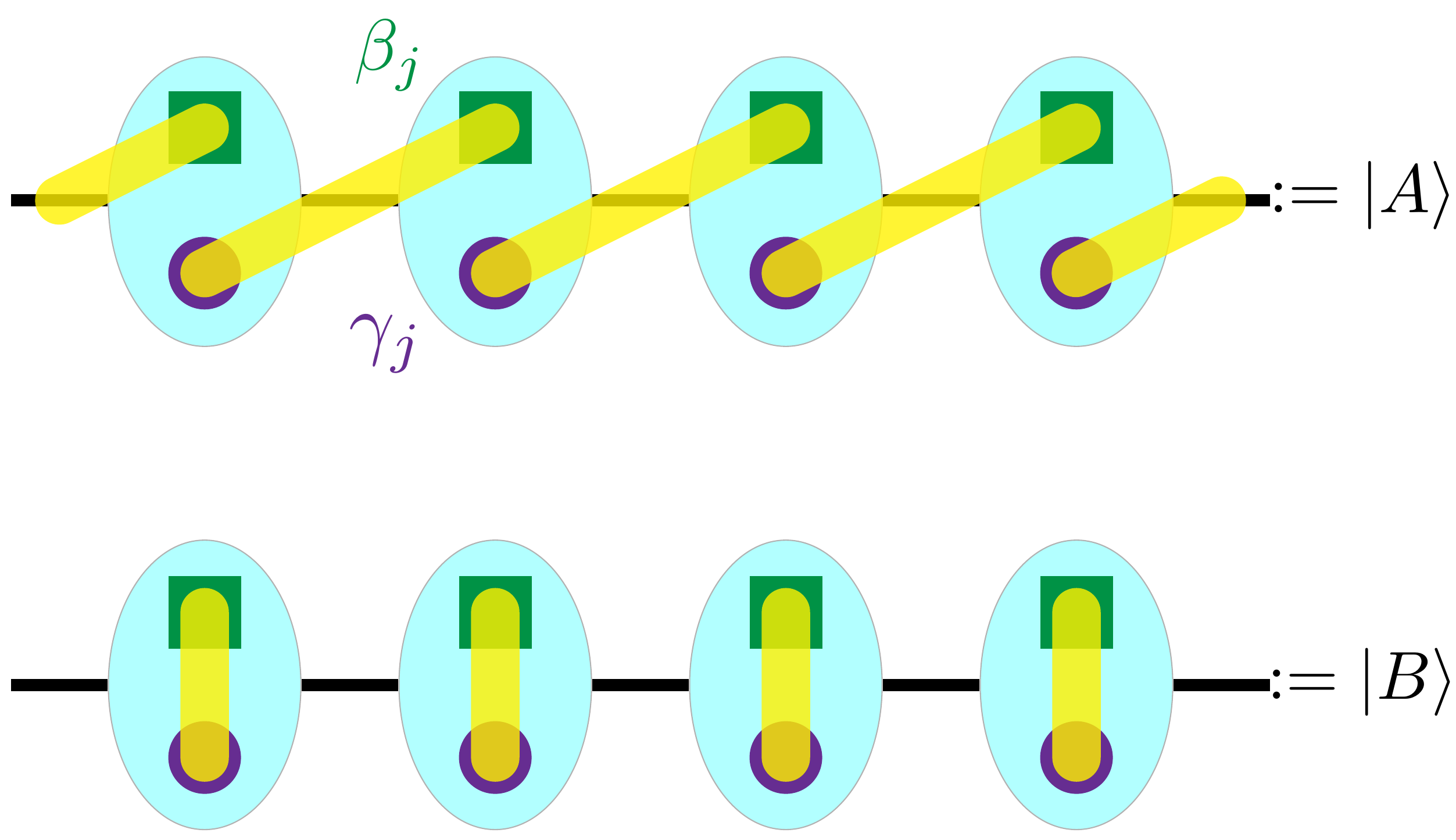}
  \caption{%
 Two degenerate zero-energy ground states of $H\qty(\alpha=\hf)$ on a periodic chain. 
  The colored bonds mean that the pair of Majorana fermions are in the state with $\im \beta_{j+1}\gamma_j= -1$ ($\im \beta_{j}\gamma_j=+1$) 
  in $\ket{\text{A}}$ ($\ket{\text{B}}$).   In the context of the Kitaev chain,  
  the states $\ket{\text{A}}$ and $\ket{\text{B}}$ are called ``topological'' and ``trivial,'' respectively. 
\label{two-GS-PBC}}
\end{figure}
%%%%%%%%%%%%%FIG%%%%%%%%%%%%%%%%%%

%%%%%%%%%%%%%%%%%%%%%%%%%%%%%%%%%%%%%%%%%%%%%%%%%%%%%%%%%%
\section{Mean-Field Approach}
\label{sec:maenfield}
%%%%%%%%%%%%%%%%%%%%%%%%%%%%%%%%%%%%%%%%%%%%%%%%%%%%%%%%%%
%%%%%%%%%%%%%%%%%%%%%%%%%%%%%%%%%%%%%%%%%%%%%%%%%%%%%%%%%%
In the previous section, we have rigorously established that the supersymmetric lattice model \eqref{eq:model} exhibits 
a transition at $\alpha=\alpha_{\text{c}}$ $\qty(\pi/8 \leq \alpha_{\text{c}}\leq \hf)$ from a SUSY-broken phase for small-$\alpha$ 
to a symmetric one for large-$\alpha$.  
In this section, we try to get more insight into the transition and the phase structure using a simple mean-field theory.  
To this end, we will first describe the model \eqref{eq:model} using the superfield formalism. By employing the superfield formalism, 
the order parameters for detecting SUSY breaking are automatically introduced.  Assuming that these order parameters 
are spatially homogeneous, we can introduce an effective Hamiltonian that allows us to automatically describe supersymmetry breaking 
and the associated transition within the mean-field framework, and find the excitations corresponding to the NG (Nambu-Goldstone) fermions.
%%%%%%%%%%%%%%%%%%%%%%%%%%%%%%%%%%%%%%%%%%%%%%%%%%%%%%%%%%%
\subsection{Superfield formalism}
%%%%%%%%%%%%%%%%%%%%%%%%%%%%%%%%%%%%%%%%%%%%%%%%%%%%%%%%%%%
Using a real Grassmann number, $\theta^2=0$, and the time variable $t$, we first introduce the supercharges $\mathcal{Q}$ and the supercovariant derivative 
$\mathcal{D}$:
\begin{align}
    \mathcal{Q}&=\frac{\partial}{\partial\theta}+\im \theta \frac{\partial}{\partial t} \\
    \mathcal{D}&=\frac{\partial}{\partial\theta}-\im \theta \frac{\partial}{\partial t},
\end{align}
which satisfy $\mathcal{Q}^2=-\mathcal{D}^2=\im \del{}{t}$ and the anticommutation relation $\acom{\mathcal{Q}}{\mathcal{D}}=0$.
Next, we define the Majorana superfields $\Gamma_j$ and $B_j$ as \cite{fu2017supersymmetric,li2017supersymmetric}:
\begin{equation}
\begin{split}
  B_j:&=\beta_j+\theta\sqrt{2} F_j \\
  \Gamma_j:&=\gamma_j+\theta\sqrt{2} G_j,
\end{split}
\end{equation}
where $F_j$ and $G_j$ act as real bosonic auxiliary fields. Using these, we can construct the action corresponding to our model \eqref{eq:model} as follows:
\begin{equation}
\begin{split}
 & S\qty[\alpha;\beta,\gamma,F,G] \\
 &  \coloneqq\frac{1}{4}\int \rmm{d}t \int \rmm{d}\theta 
  \sum_j\qty(B_j\mathcal{D} B_j+\Gamma_j \mathcal{D}\Gamma_j)    \\
& \phantom{\coloneqq} 
- \frac{1}{\sqrt{2}}\int \rmm{d}t\int \rmm{d}\theta\sum_j\qty{
    \qty(1-\alpha)B_j+ \im \alpha B_j B_{j+1}\Gamma_j} \\
 & =\int \rmm{d}t \sum_{j=1}^L \Biggl[
    \frac{1}{4} \beta_j \im \dot{\beta}_j+\frac{1}{4} \gamma_j \im \dot{\gamma}_j+\frac{1}{2} F_j^2+\frac{1}{2} G_j^2   \\
 &  \phantom{=}
 - F_j\qty{\qty(1-\alpha)+ \alpha 
     \qty(\im \beta_{j+1}\gamma_j-\im \beta_{j-1} \gamma_{j-1} )}-\alpha G_j \im \beta_j \beta_{j+1}
  \Biggr].
  \end{split}
  \label{eq:action_off}  
\end{equation}
Indeed, after eliminating the auxiliary fields in \eqref{eq:action_off} and performing the Legendre transformation, 
we recover the original Hamiltonian \eqref{eq:model} (up to a constant). The canonical commutation relations to be imposed on the Majorana fermions $\beta_j$, $\gamma_j$, and their conjugate momenta $p^{(\beta)}_j=\frac{\im}{4}\beta_j$, $p^{(\gamma)}_j=\frac{\im}{4}\gamma_j$ are:
\begin{align}
  &\acom{\beta_i}{p^{(\beta)}_j}=\acom{\gamma_i}{p^{(\gamma)}_j}=\frac{\im}{2} \delta_{i,j} \nt
  &\qty(=0 \text{ otherwise}) \; .
\end{align}

%%%%%%%%%%%%%%%%%%%%%%%%%%%%%%%%%%%%%%%%%%%%%%%%%%%%%%%%%

The equations of motion for the auxiliary fields, derived from \eqref{eq:action_off}, are given by:
\begin{equation}
\begin{split}
&F_j=\qty(1-\alpha)+ \alpha
\qty(\im \beta_{j+1}\gamma_j-\im \beta_{j-1} \gamma_{j-1})=\frac{1}{2} \acom{\beta_j}{R\qty(\alpha)} \\ 
&G_j=\alpha \im \beta_j \beta_{j+1}=\frac{1}{2} \acom{\gamma_j}{R\qty(\alpha)}  \; .
\end{split}
\label{eq:orderparameter}
\end{equation}
Here, $F_j$ and $G_j$ function as the order parameters for SUSY breaking \cite{wess1992supersymmetry}. 
This becomes evident when considering that in the absence of supersymmetry breaking, the ground state $\ket{0}_0$ of the system is a zero-energy eigenstate that annihilates under $R\qty(\alpha)$. Therefore, from \eqref{eq:orderparameter}, we have:
\begin{equation}
\begin{split}
&{}_0\!\bra{0}F_j\ket{0}_0=\frac{1}{2} {}_0\!\bra{0}\acom{\beta_j}{R\qty(\alpha)}\ket{0}_0=0  \\
&{}_0\!\bra{0}G_j\ket{0}_0=\frac{1}{2} {}_0\!\bra{0}\acom{\gamma_j}{R\qty(\alpha)}\ket{0}_0=0  \; .
\end{split}
\end{equation}
These considerations are relevant to finite systems and provide order parameters. However, in the case of infinite systems, the breaking of supersymmetry is defined with $e_{\text{g.s.}} > 0$, and it is important to note that SUSY unbroken does not necessarily imply the existence of the zero-energy ground state $\ket{0}_0$. Therefore, these order parameters may lose their meaning in infinite systems. However, if these spatial averages are defined, they can still detect the breaking of supersymmetry. Detailed discussions on this matter are provided in Appendix \ref{App:orderineq}.
Furthermore, for supersymmetry doublets composed of states with positive energy $E>0$, the expectation values of the order parameters are the same 
(see Appendix~\ref{App:orderparameter}).

Therefore, retaining the auxiliary fields $F_j$ and $G_j$, and further assuming that they are uniform, we obtain the following mean-field Hamiltonian:
\begin{equation}
\label{eq:mfHam}
\begin{split}
H_{\text{mf}}\qty(\alpha;F,G) \coloneqq & \qty(\qty(1-\alpha)F-\frac{1}{2} F^2-\frac{1}{2} G^2)L  \\
& + \alpha F\sum_{j=1}^L \im \qty(\beta_{j+1}- \beta_{j}) \gamma_{j}+\alpha G\sum_{j=1}^{L}\im \beta_j \beta_{j+1}  \; .
\end{split}
\end{equation}
This Hamiltonian includes the order parameters $F$ and $G$ and can be interpreted as a mean-field theory to detect supersymmetry breaking.

%%%%%%%%%%%%%%%%%%%%%%%%%%%%%%%%%%%%%%%%%%%%%%%%%%%
\subsection{SUSY breaking in mean-field theory}
%%%%%%%%%%%%%%%%%%%%%%%%%%%%%%%%%%%%%%%%%%%%%%%%%%%
The energy density of the ground state, derived from \eqref{eq:mean-field-Edensity}, is given by:
\begin{equation}
\begin{split}
e^{\text{mf}}_{\text{g.s.}}\qty(\alpha;F,G) =& \qty(1-\alpha) F-\frac{1}{2}F^2-\frac{1}{2}G^2  \\
& - \frac{2\alpha}{\pi}\qty(\sqrt{F^2+G^2}+\frac{F^2}{\abs{G}}\sinh^{-1}\qty(\abs{\frac{G}{F}}))  \; .
\end{split}
\end{equation}
The values of $F$ and $G$ that realize the extremum of the energy density are given by:
\begin{subequations}
  \begin{align}
   \label{F-MF}
  & F\qty(\alpha)=\left\{
  \begin{aligned}
  &1-\frac{4+\pi}{\pi}\alpha &&  \qty(\alpha\leq \frac{\pi}{4+\pi}) \\
  &0 && \qty(\frac{\pi}{4+\pi} \leq \alpha)
  \end{aligned}
  \right.  \\
  \label{G-MF}
  & G\qty(\alpha)=0 \\
  \label{E-MF}
  &e^{\rmm{mf}}_{\text{g.s.}}\qty(\alpha)=
  \left\{
  \begin{aligned}
  & \hf \qty(1-\frac{4+\pi}{\pi}\alpha)^{2} &&  \qty(\alpha\leq \frac{\pi}{4+\pi}) \\
  &0 &&\qty(\frac{\pi}{4+\pi} \leq \alpha)  \; .
  \end{aligned}
  \right.    
  \end{align}
  \end{subequations}
Thus, the mean-field theory predicts the critical point $\alpha_\text{c}^\text{mf}=\frac{\pi}{4+\pi}=0.439900\ldots$  
(see Fig.~\ref{Fig-Maj-SUSY-meanfield-vs-num}).

%%%%%%%% FIG %%%%%%%%%%%%%%%%%%%%%%%%%%%%%%%%%%%%%%%%%%

\begin{figure}[htbp]
  \includegraphics[width=\columnwidth,clip]{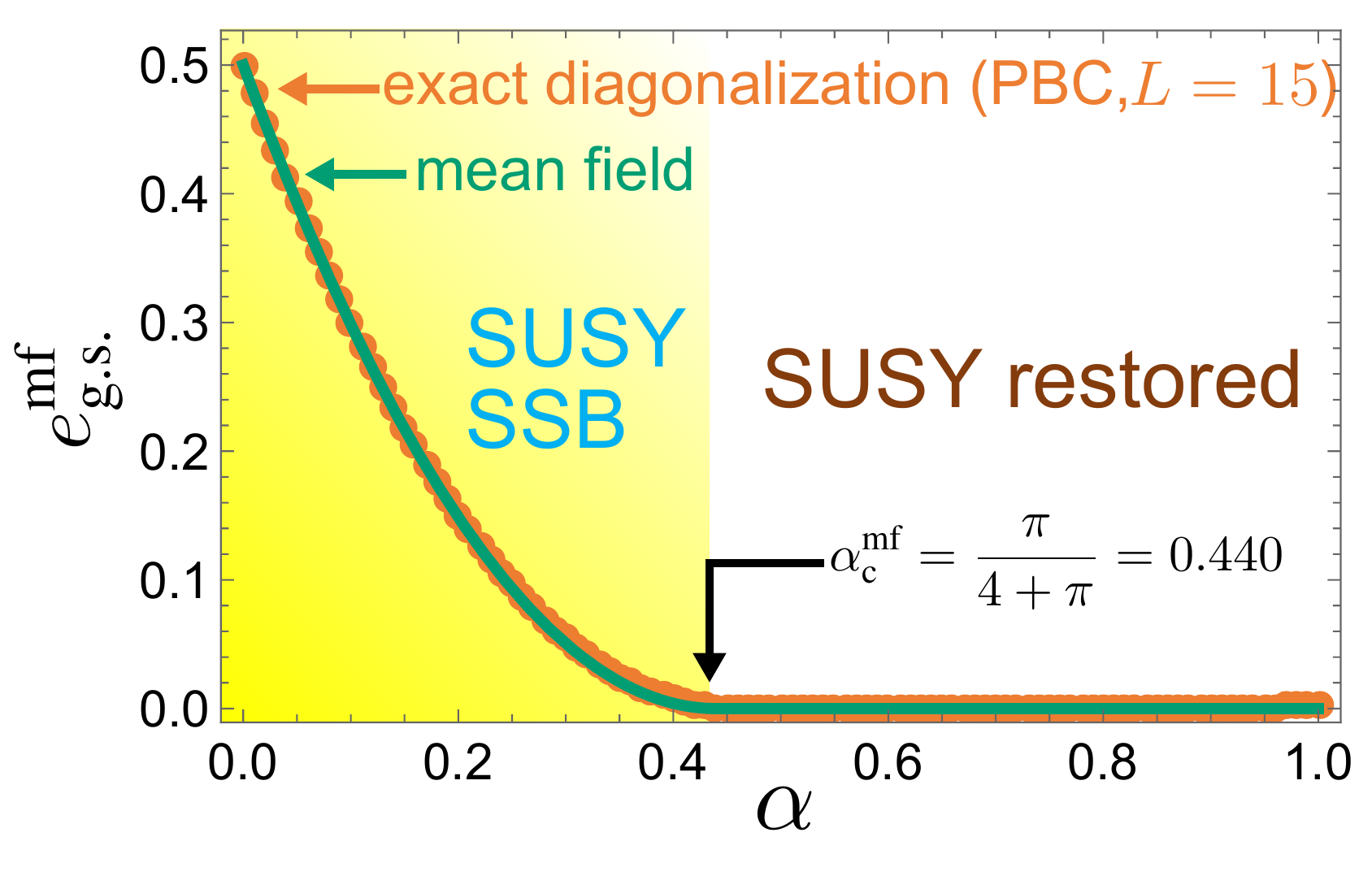}
  \caption{
The ground-state energy per site as a function of $\alpha$. The solid green curve represents the results of mean-field theory \eqref{E-MF}, while the orange dots depict the outcomes of exact diagonalization for a system with periodic boundary conditions and a size of $L=15$. Mean-field theory tends to slightly underestimate the actual values but is quite close to the results of exact diagonalization. Furthermore, the mean-field theory's transition point at $\alpha_{\text{c}}^{\text{mf}}=\frac{\pi}{4+\pi}$, within the scope of mean-field theory, is consistent with the results from Sec.~\ref{sec:SUPERSYMMETRY BREAKING}, where $\frac{\pi}{8}\leq \alpha_{\text{c}}^{\text{mf}}=\frac{\pi}{4+\pi}\leq \hf$. 
  \label{Fig-Maj-SUSY-meanfield-vs-num}}
\end{figure}
%%%%%%%%%%%%%%%%%%%%%%%%%%%%%%%%%%%%%%%%%%%%%%%%%%%%%%%

We would like to comment on the meaning of our ``mean-field theory'' here. In the path-integral formulation of the standard mean-field theories (e.g., the BCS theory), auxiliary fields as the variational parameters ($F$ and $G$ here) are first introduced through the Hubbard-Stratonovich transformation. If we integrate out the fermion fields, we are left with the effective action of $F$ and $G$ which is to be integrated over these auxiliary fields to yield the final result. Mean-field approximation corresponds to trading the functional integral over $F$ and $G$ for the saddle point value (i.e., the extremum) of the effective action. Our treatment in Sec.~\ref{sec:maenfield} follows this strategy. 
Therefore, the optimal values \eqref{F-MF} and \eqref{G-MF} are not always {\em minimizing} the (mean-field) ground-state energy within the space of a specific variational wave function. 

For comparison, we also performed variational calculations using the trial ground state of the free Majorana fermions with a tunable mass parameter as in Ref.~\cite{aasen2020electrical}. From the optimal value of the Majorana mass parameter, we obtained the critical value $\alpha_{\text{c}}=0.392$ at which the Majorana excitation as the NG fermion starts having a finite gap. The above value $\alpha_{\text{c}}^{\text{mf}}\approx0.44$ from the superfield mean-field approximation is closer to the preliminary numerical result $\alpha_{\text{c}}=0.46$~\cite{miura_totsuka_unpublished}. This may suggest that our mean-field approach based on the superfield gives a better approximation than the simple-minded variational calculations.

%%%%%%%%%%%%%%%%%%%%%%%%%%%%%%%%%%%%%%%%%%%%%%%%%%%
\subsection{NG fermion in mean-field theory}
\label{sec:NG-fermion-in-MFA}
%%%%%%%%%%%%%%%%%%%%%%%%%%%%%%%%%%%%%%%%%%%%%%%%%%%
By substituting \eqref{F-MF} and \eqref{G-MF} into \eqref{eq:dipersionMF}, we see that there is a dispersive fermionic excitation 
in the SUSY-broken phase $0 \leq \alpha\leq \alpha_\text{c}^\text{mf}=\frac{\pi}{4+\pi}$
having the following single-particle spectrum: 
\begin{equation}
\label{eq:dispMF}
\begin{split}
&\eps_{\rmm{mf}}\qty(k)
=4\alpha\abs{F(\alpha)}\abs{\sin(\frac{k}{2})}
= \frac{4\alpha}{\alpha^\text{mf}_{\text{c}}} \qty(\alpha^\text{mf}_{\text{c}} - \alpha)\abs{\sin(\frac{k}{2})} \\
& \qty(-\pi < k \leq \pi) \; .
\end{split}
\end{equation}
For small values of $k$, it exhibits a gapless linear dispersion (Fig.~\ref{Fig-NG-dispersion}), 
whereas, in the SUSY-symmetric phase $\alpha\geq\frac{\pi}{4+\pi}$, $F(\alpha)=0$, and the fermionic excitation loses 
its dispersion. 
This behavior supports the hypothesis that this elementary excitation is the NG fermion and 
is also consistent with the long-wavelength behavior predicted by the single-mode approximation for the NG fermions:
\begin{align}
\epsilon_{\text{var}}(k) \leq \frac{1}{2}\sqrt{\frac{C}{E_{\text{g.s.}}/L}}\left| k \right| + O\left(\left| k \right|^3\right) 
\end{align}
(see Appendix~\ref{App:NG-single} for the derivation).   

%%%%%%%% FIG %%%%%%%%%%%%%%%%%%%%%%%%%%%%%%%%%%%%%%%%%%

\begin{figure}[htbp]
  \includegraphics[width=\columnwidth,clip]{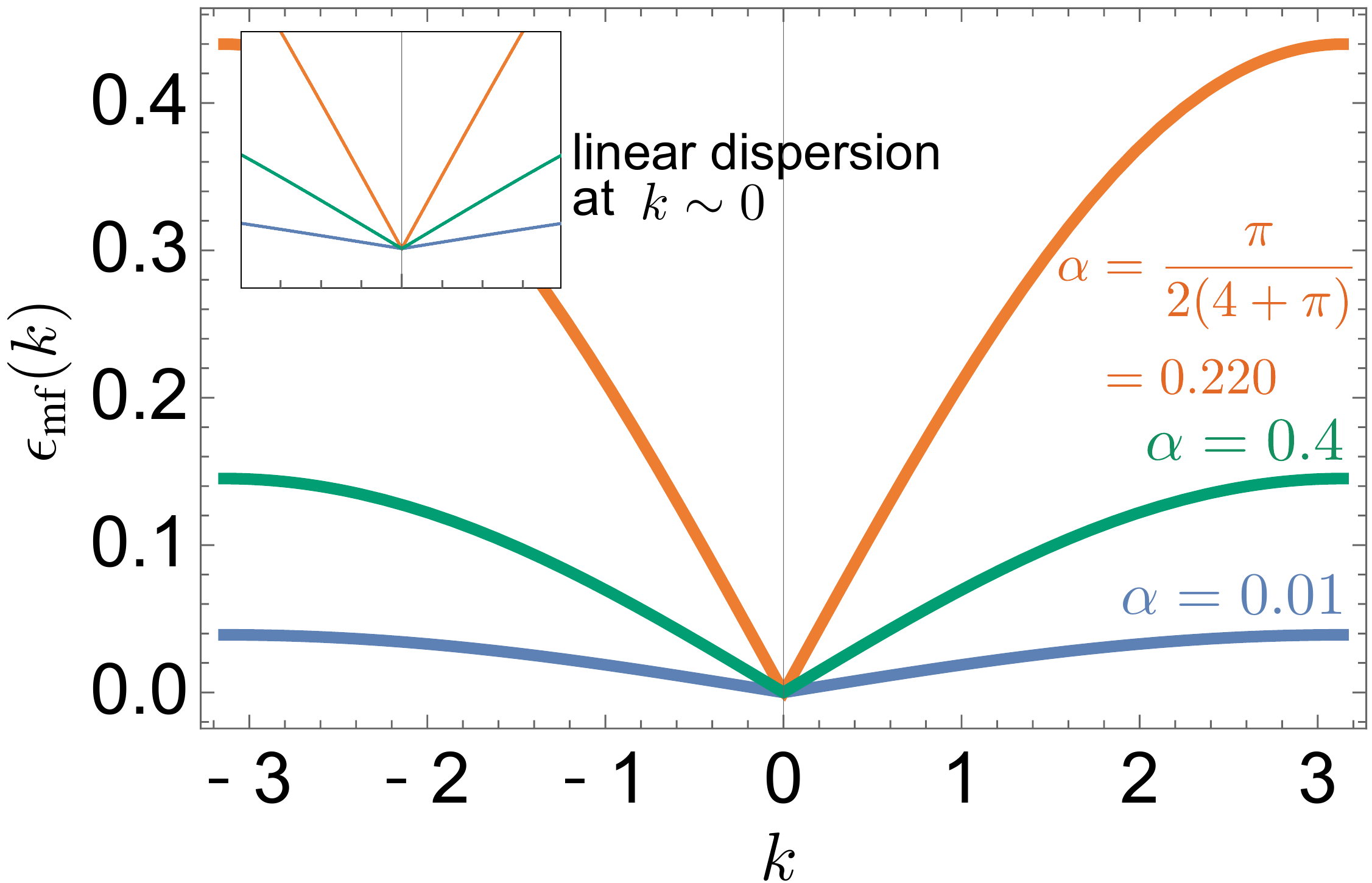}
  \caption{
  The mean-field spectrum \eqref{eq:dispMF} of the elementary excitation for several values of \(\alpha\).   
  When $\alpha$ is increased from $0$, the dispersion first becomes larger and takes its maximum at \(\alpha= \alpha_\text{c}^\text{mf}/2\).  Then, the fermion gradually loses its dispersion towards \(\alpha=\alpha_\text{c}^\text{mf} \) at which it becomes dispersionless again. 
 As the fermionic excitation has a finite dispersion only when SUSY is broken $\qty(0 \leq \alpha < \frac{\pi}{4+\pi})$, it can be identified as the NG fermion.
  \label{Fig-NG-dispersion}}
\end{figure}
%%%%%%%%%%%%%%%%%%%%%%%%%%%%%%%%%%%%%%%%%%%%%%%%%%%%%%%

%%%%%%%%%%%%%%%%%%%%%%%%%%%%%%%%%%%%%%%%%%%%%%%%%%%%%%%%%%
%%%%%%%%%%%%%%%%%%%%%%%%%%%%%%%%%%%%%%%%%%%%%%%%%%%%%%%%%%
\section{Kink and Skink as zero-energy defects}
\label{skinksection}
%%%%%%%%%%%%%%%%%%%%%%%%%%%%%%%%%%%%%%%%%%%%%%%%%%%%%%%%%%
%%%%%%%%%%%%%%%%%%%%%%%%%%%%%%%%%%%%%%%%%%%%%%%%%%%%%%%%%%%
Having obtained a rough picture of the phase structure and low-lying excitations, we will now proceed to closely investigate 
the nature of the ground states and low-energy spectrum.   
To keep the clarity of the argument and avoid the complexity coming from the nonlocal edge fermions, 
we will use below the spin 1/2 language introduced in Sec.~\ref{sec:spin-rep}, in which 
a pair of Majorana fermions at each physical site is mapped to a spin 1/2 at the same site by the Jordan-Wigner transformation (see Fig.~\ref{fig:01}). 
Accordingly, the Hamiltonian is given by \eqref{eq:spin-rep-OBC} (for OBC) or \eqref{eq:spin-rep-PBC} (for PBC).

%%%%%%%%%%%%%%%%%%%%%%%%%%%%%%%%%%%%%%%%%%%%%%%%%%%%%%%%%%%%%%%%%%%%
\subsection{Numerical spectrum}
\label{sec:numerical-spec}
%%%%%%%%%%%%%%%%%%%%%%%%%%%%%%%%%%%%%%%%%%%%%%%%%%%%%%%%%%%%%%%
Numerical calculations for the spectrum of the Hamiltonian \eqref{eq:model} under the periodic and open boundary conditions 
revealed several interesting results. In Fig.~\ref{Fig-OBCgraphs}, we plot the lowest 50 eigenvalues obtained by exact diagonalizations 
for a system size $L=15$ system at various values of $\alpha$ ($=0.2$, $0.48$, $0.5$, $0.52$, $0.8$, and $1$).  
The data for the periodic (open) system are plotted by blue (orange) dots.   
The spectral structure for the smallest $\alpha=0.2$ [Fig.~\ref{Fig-OBCgraphs}(a)] can be easily understood 
by the spectrum of $\alpha\qty(1-\alpha)H_{\text{CTI}}$ since the interaction term $\alpha^2H_{\text{int}}$ is sufficiently small. 
For instance, the spectral degeneracies $2$ (ground state), $4$, $2$, $4$, $4$, ...(for PBC) and $2$ (ground state), $2$, $2$, $4$, $4$, ...(for OBC) of some lowest-lying levels are fully understood by 
those of the free massless Majorana fermion. This, together with the gapless $k$-linear spectrum of the NG fermion found in the mean-field analysis, strongly suggests that the small-$\alpha$ region is a critical phase belonging to the $c=1/2$ critical Ising (i.e., gapless Majorana) universality class.  

On the other hand, for $\alpha\approx \hf$, a clear large plateau structure is seen in the low-energy part of the spectrum for both PBC and OBC 
[Fig.~\ref{Fig-OBCgraphs}(b)-(d)], which was absent when $\alpha=0.2$.  
At $\alpha=0.48$, the left-hand side of the plateau is still a bit rounded [Fig.~\ref{Fig-OBCgraphs}(b)], while, at $\alpha=\hf$, 
it becomes completely flat (for OBC) indicating that 
there are $2L\,(=30)$ degenerate ground states with the energy $E=\frac{1}{8}$ [Fig.~\ref{Fig-OBCgraphs}(c)].   
At $\alpha=0.52$, now the right-hand side starts having a finite slope [Fig.~\ref{Fig-OBCgraphs}(d)], and at $\alpha=0.8$, 
the plateau structure is already obscured entirely [Fig.~\ref{Fig-OBCgraphs}(e)].   
Except at $\alpha =0, \hf$ and $1$, the ground-state degeneracy is always $2$ both for PBC and for OBC regardless of 
whether SUSY is broken or not.  
This is just a natural consequence of the supersymmetric spectrum; 
as any positive-energy levels come in pairs, there must be an even number (two, in generic cases) of zero-energy states, if they exist (note that 
the entire Hilbert space is $2^{L}$-dimensional).  
Last, at $\alpha=1$, the spectrum exhibits extensive degeneracy in each energy level, which may be understood as 
the signature of ``superfrustration'' as discussed in Ref.~\cite{fendley2019free}.   

A remark is in order about SUSY SSB at $\alpha=1$. From the positive ground-state energy in Fig.~\ref{Fig-OBCgraphs} (f), one may think that SUSY is broken again at $\alpha=1$. However,  considering the exact result that SUSY is restored at $\alpha=1$ under OBC \cite{fendley2019free} and the boundary-condition independence of the criterion \eqref{eqn:SSB-criterion} (see Appendix \ref{App:specvary}), it is suggested that, in an infinite-size system, SUSY SSB occurs neither for OBC nor for PBC. In fact, we numerically calculated the ground-state energy density $e_{\text{g.s.}}(\alpha=1)$ for increasing system sizes $L$ to confirm that it converges to zero algebraically, implying unbroken SUSY in the sense of the criterion \eqref{eqn:SSB-criterion}.

In summary, we have numerically observed that around $\alpha=\hf$ a plateau structure forms in the low-energy spectrum and 
that the ground-state degeneracy for OBC changes from $2$ (for generic $\alpha$) to $2L$ precisely at $\alpha=\hf$, 
whereas no such change is seen for PBC, except for $\alpha=0,1$.
In Sec.~\ref{skink}, we will show that this plateau structure can be attributed to the existence of SUSY doublets (``kink'' and ``skink'') of low-energy defects 
around $\alpha=\hf$.

%%%%%%%% FIG %%%%%%%%%%%%%%%%%%%%%%%%%%%%%%%%%%%%%%%%%%
\begin{figure}[htbp]
  \includegraphics[width=\columnwidth,clip]{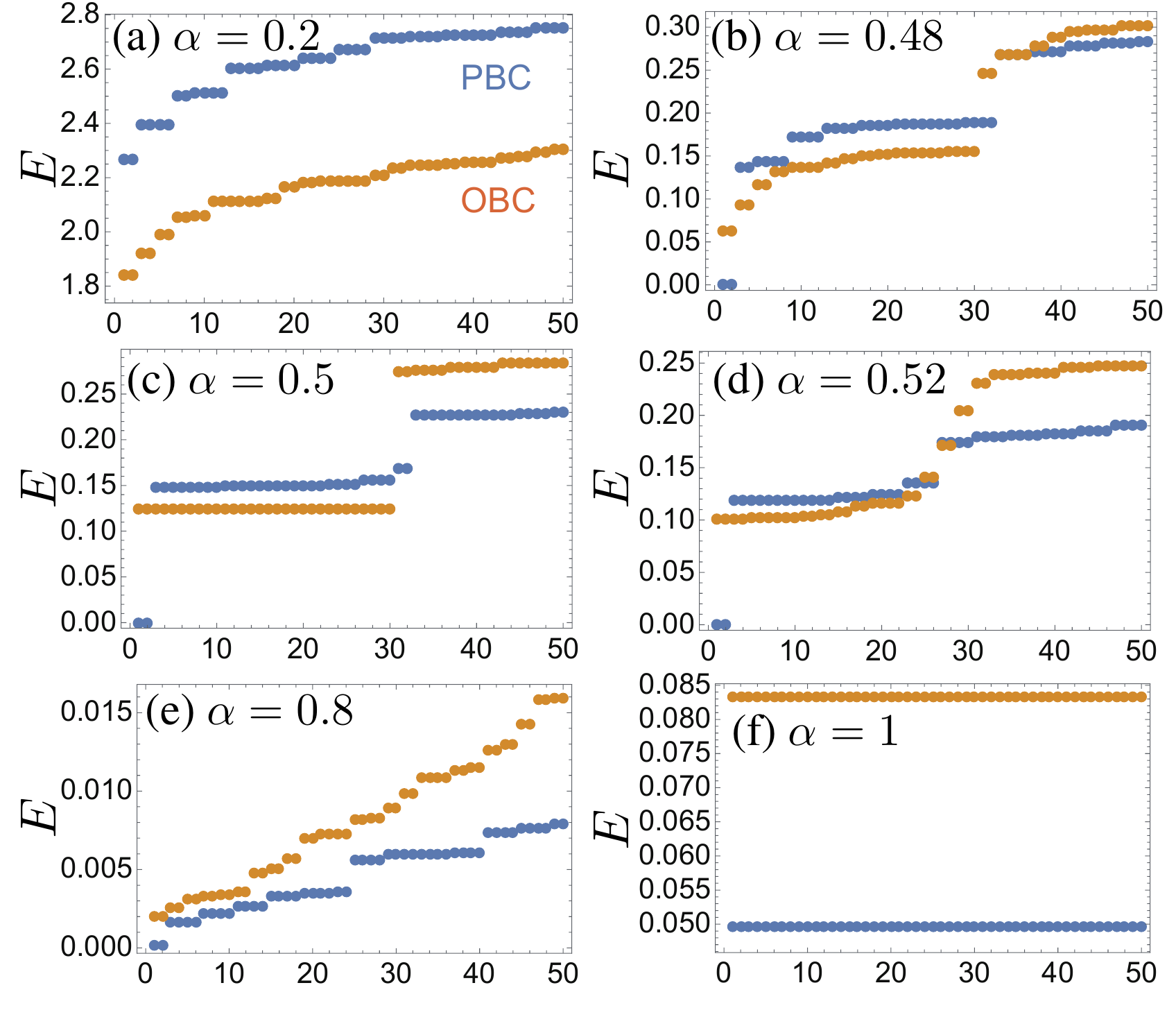}
  \caption{
The change in the spectrum of $H(\alpha)$ on finite ($L=15$) open (orange dots) and periodic (blue dots) chains when $\alpha$ is varied: 
(a) \(\alpha=0.2\), (b) \(0.48\), (c) \(0.5\), (d) \(0.52\), (e) \(0.8\), and (f) \(1\).   
Only the lowest 50 eigenvalues are shown.  
Around \(\alpha=1/2\), a noticeable flat structure emerges in the spectrum [(b), (c), and (d)].  
At \(\alpha=1/2\), it becomes completely flat for an open chain indicating that there are $2L\,(=30)$ degenerate ground states.  
When \(\alpha=0.8\), the flat structure is already obscured and the spectrum exhibits a qualitatively different structure.
Finally at \(\alpha=1\), the spectrum shows extensive degeneracy which is a hallmark of superfrustration that occurs at \(\alpha=1\) \cite{fendley2019free}. 
For generic values of $\alpha$, the ground states are twofold degenerate for both PBC and OBC.}  

\label{Fig-OBCgraphs}
\end{figure}
%%%%%%%%%%%%%%%%%%%%%%%%%%%%%%%%%%%%%%%%%%%%%%%%%%%%%%%

%%%%%%%%%%%%%%%%%%%%%%%%%%%%%%%%%%%%%%%%%%%%%%%%%%%
\subsection{Kink and skink}
\label{skink}
%%%%%%%%%%%%%%%%%%%%%%%%%%%%%%%%%%%%%%%%%%%%%%%%%%%
In the previous section, we have numerically observed that the ground states of an {\em open} chain are $2L$-fold degenerate at $\alpha=\frac{1}{2}$, 
while  for a periodic chain, there are exactly two zero-energy ground states, which may be interpreted as the trivial phase (quantum paramagnetic phase in the spin language) and topological phases (ferromagnetic phase) of 
the Kitaev chain (transverse field Ising model) \cite{kitaev2001unpaired}.

Remarkably, the existence of these $O(L)$-degenerate 
{\em ground} 
states on an open chain can be understood by 
constructing a domain wall connecting the topological and trivial states found in a periodic chain. 
Specifically, in the Jordan-Wigner transformed spin language, it can be shown that the following states (Fig.~\ref{fig:07}):
\begin{equation}
\begin{split}
&\ket{j,\migi} \coloneqq\ket{\dw\cdots \dw
\underset{j}{\up}
\migi \cdots\migi} ,  \\
&\ket{j,\leftarrow}\coloneqq\ket{\dw\cdots\dw
\underset{j}{\up}
\leftarrow\cdots\hdr}   \\
&\qty(j=1,\cdots,L-1),   \\
&\ket{\dw\cdots \dw\up} \, ,  \; \ket{\dw\cdots \dw\dw} \, (=\ket{\text{B}}) 
\end{split}
\label{eq:migi&hidari}
\end{equation}
are the $2L$-fold degenerate exact ground states with the $L$-{\em independent} eigenvalue $\frac{1}{8}$ at $\alpha=\hf$. 
Here, $\migi$ ($\leftarrow$) represents the eigenstate of the Pauli matrix $\sx$ with eigenvalue $+1$ ($-1$) and 
labels the two degenerate ferromagnetic states of the Ising chain.  
%%%
To see that these are all the ground states with the energy $\frac{1}{8}$, we note that  
the Hamiltonian takes the following simple form when $\alpha=\hf$:
\begin{subequations}
\begin{equation}
H\qty(\alpha=\hf) = \frac{1}{8} + \sum_{j=1}^{L-2} h_{j}   
\end{equation}
with the local Hamiltonian $h_{j}$ being a projection operator:
\begin{equation}
h_j\coloneqq\frac{1}{4}\qty(1+\sz_{j})\qty(1-\sx_{j+1}\sx_{j+2})\, (\geq 0)  \; .
\label{eqn:local-Ham-by-spin}
\end{equation}
\end{subequations}
Then, it is easy to check that each local Hamiltonian $h_{j}$ annihilates all the $2L$ states in Eq.~\eqref{eq:migi&hidari}.    
For example, the following ten orthogonal states span the exact ground-state subspace (with energy $\frac{1}{8}$) of an $L=5$ open chain at $\alpha=\hf$:
\begin{equation}
\begin{split}
& \ket{\up\migi\migi\migi\migi} \,(=\ket{1,\migi}), \;   \ket{\up\hdr\hdr\hdr\hdr} \,(=\ket{1,\hdr}), \\
& \ket{\dw\up\migi\migi\migi} \,(=\ket{2,\migi}), \;\; \ket{\dw\up\hdr\hdr\hdr} \,(=\ket{2,\hdr}) , \\
&\ket{\dw\dw\up\migi\migi} \,(=\ket{3,\migi}), \;\; \ket{\dw\dw\up\hdr\hdr} \,(=\ket{3,\hdr}), \\
& \ket{\dw\dw\dw\up\migi} \,(=\ket{4,\migi}),  \;\;  \ket{\dw\dw\dw\up\hdr} \,(=\ket{4,\hdr}), \\
& \ket{\dw\dw\dw\dw\up}, \;\;   \ket{\dw\dw\dw\dw\dw} \; .
  \end{split}
\end{equation}
Taking into account the double degeneracy ($\ket{\migi\migi\cdots}$ and $\ket{\hdr\hdr\cdots}$) in the ferromagnetic region (to the right of the domain wall) as well as the $L$ different choices of the domain wall's location, we see that these domain-wall ground states exhibit the desired $2L$-fold degeneracy.

Here we would like to comment on the ground-state energy $\frac{1}{8}$ of the above ``domain-wall'' states. By the analogy to the usual domain walls in ordered systems (like the ferromagnetic state in the Ising model), one may think that it corresponds to the local energy increase arising from the domain-wall formation. However, the constant $\frac{1}{8}$ is an artifact of the construction of the supersymmetric Hamiltonian \eqref{eq:model} and does not have any physical meaning. 
In fact, the homogeneous states $\ket{\text{A}}$ and $\ket{\text{B}}$ without domain walls (see Sec.~\ref{sec:SUSYrestoring}), which are equally eligible in the sense that they optimize {\em all} the local terms $h_{j}$ in the Hamiltonian, are also included in the above domain-wall subspace with energy $\frac{1}{8}$. Therefore, despite their positive energy, these domain-wall states are created at {\em zero} energy cost. 

It is interesting to rephrase the above in the original fermion language.   
To this end, it is essential to note that the pair of ``partially-ferromagnetic'' states $\ket{j,\migi}$ and $\ket{j,\leftarrow}$ 
are not the eigenstates of the fermion parity operator 
$\qty(-1)^F$.  In fact, using  Eq.~\eqref{eqn:F-parity-spin}, we have:
\begin{equation}
\begin{split}
&\qty(-1)^F\ket{j,\migi}=\qty(-1)^{L-j+1}\ket{j,\leftarrow}  \\
&\qty(-1)^F\ket{j,\leftarrow}=\qty(-1)^{L-j+1}\ket{j,\migi} \\
&\qty(-1)^F\ket{\dw\cdots \dw\up}=-\ket{\dw\cdots \dw\up} \\
&\qty(-1)^F\ket{\dw\cdots\dw\dw}=\ket{\dw\cdots \dw\dw}  \; .
\end{split}
\end{equation}
Hence, we see the following linear combinations:  
\begin{align}
&\ket{j,\rmm{kink}}\nt
&\coloneqq
\left\{
\begin{aligned}
&\frac{1}{\sqrt{2}}\qty(\ket{j,\hdr}-
\qty(-1)^{L-j+1}\ket{j,\migi}) && (j=1,\cdots,L-1) \\
&\ket{\dw\cdots \dw\up} && (j=L)
\end{aligned}
\right.
\nt
&\ket{j,\rmm{skink}}\nt
&\coloneqq
\left\{
\begin{aligned}
&\frac{1}{\sqrt{2}}\qty(\ket{j,\hdr}+\qty(-1)^{L-j+1}\ket{j,\migi}) && (j=1,\cdots,L-1) \\
&\ket{\dw\cdots \dw\dw} && (j=L)
\end{aligned}
\right.
\label{eq:kink&skink}
\end{align}
are the eigenstates of the fermion parity operator $\qty(-1)^F$:
\begin{equation}
\begin{split}
    &\qty(-1)^F\ket{j,\rmm{skink}}=+\ket{j,\rmm{skink}}, \\
    &\qty(-1)^F\ket{j,\rmm{kink}}=-\ket{j,\rmm{kink}}.
\end{split}
\end{equation}
As can be easily verified, these states transform into each other under the action of $R\qty(\alpha=\hf)$ and form a SUSY doublet: 
\begin{equation}
\begin{split}
& R\qty(\alpha=\hf)\ket{j,\rmm{skink}}=\hf\ket{j,\rmm{kink}}   \\
& R\qty(\alpha=\hf)\ket{j,\rmm{kink}}=\hf\ket{j,\rmm{skink}}  \; .
\end{split}
\end{equation}

Kinks in supersymmetric lattice fermion models are discussed in, e.g., Refs.~\cite{o2018lattice,minavr2022kink,wilhelm2023supersymmetry}.
The work in Ref.~\cite{o2018lattice} is particularly interesting for its discussion of domain walls connecting ordered and disordered states, similar to our observations based on the frustration-free properties. However, our model differs in that it possesses exact 
\(\mathcal{N}=1\) supersymmetry on a lattice. What is more important from a physical perspective is that our model does not require a finite-energy cost to create a domain wall. This leads to the unique situation that both the uniform state and the state with a domain wall coexist as exact ground states, in contrast to the more conventional situation in Ref.~\cite{o2018lattice}.
%%%%%%%% FIG %%%%%%%%%%%%%%%%%%%%%%%%%%%%%%%%%%%%%%%%%%
\begin{figure}[htbp]
  \includegraphics[width=\columnwidth,clip]{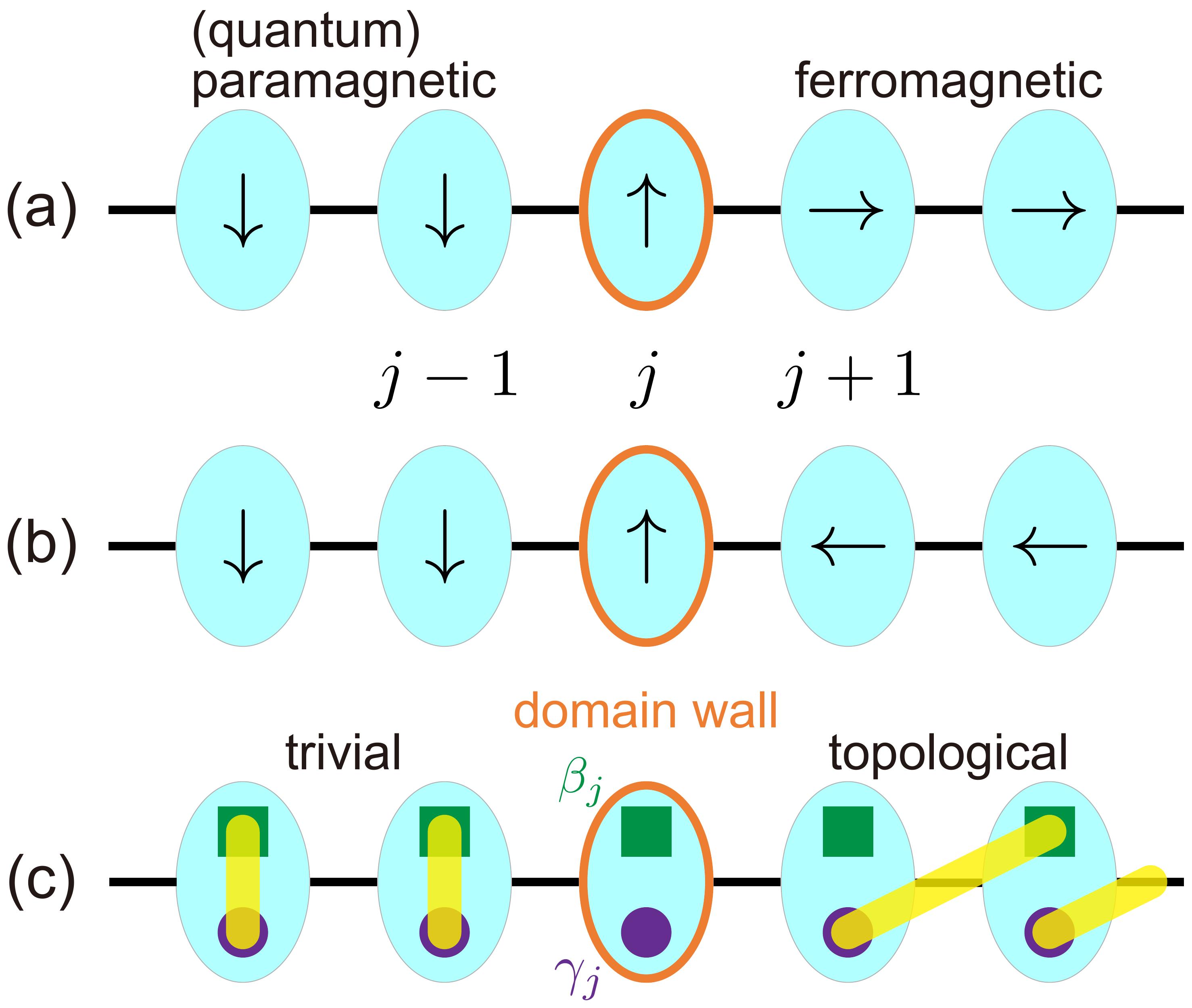}
  \caption{
 The top (a) and middle (b) panels, respectively, illustrate the spin states $\ket{j,\migi}$ and $\ket{j,\hdr}$ in \eqref{eq:migi&hidari}.   
The states $\ket{j,\text{kink}}$ and $\ket{j,\text{skink}}$ are given by linear superpositions of these.  
(c) In the language of Majorana fermions, these states may be interpreted as having the trivial and topological phases 
on the left- and right-hand sides of the Kitaev chain, respectively, with a domain wall (the $\uparrow$ spin in the spin language) 
separating them.
  \label{fig:07}}
\end{figure}
%%%%%%%%%%%%%%%%%%%%%%%%%%%%%%%%%%%%%%%%%%%%%%%%%%%%%%%

The order parameter \eqref{eq:orderparameter} at $\alpha=\hf$
\begin{equation}
F_i = \frac{1}{2} \left( 1 + \sigma_{i-1}^{z} - \sigma_{i}^{x} \sigma_{i+1}^{x} \right)  
\quad (i=2,\ldots, L-1)
\end{equation}
is nonzero only at the location of the kink, where the supersymmetry is broken locally (Fig.~\ref{Fig-fs}):
\begin{align}
  \bra{j,\rmm{(s)kink}}F_i\qty(\alpha=\hf)\ket{j,\rmm{(s)kink}}=\left\{
  \begin{aligned}
  &0 && \qty(i\neq j+1) \\
  &\hf && \qty(i=j+1)  \; .
  \end{aligned}
  \right.
\end{align}

%%%%%%%% FIG %%%%%%%%%%%%%%%%%%%%%%%%%%%%%%%%%%%%%%%%%%

\begin{figure}[htbp]
  \includegraphics[width=\columnwidth,clip]{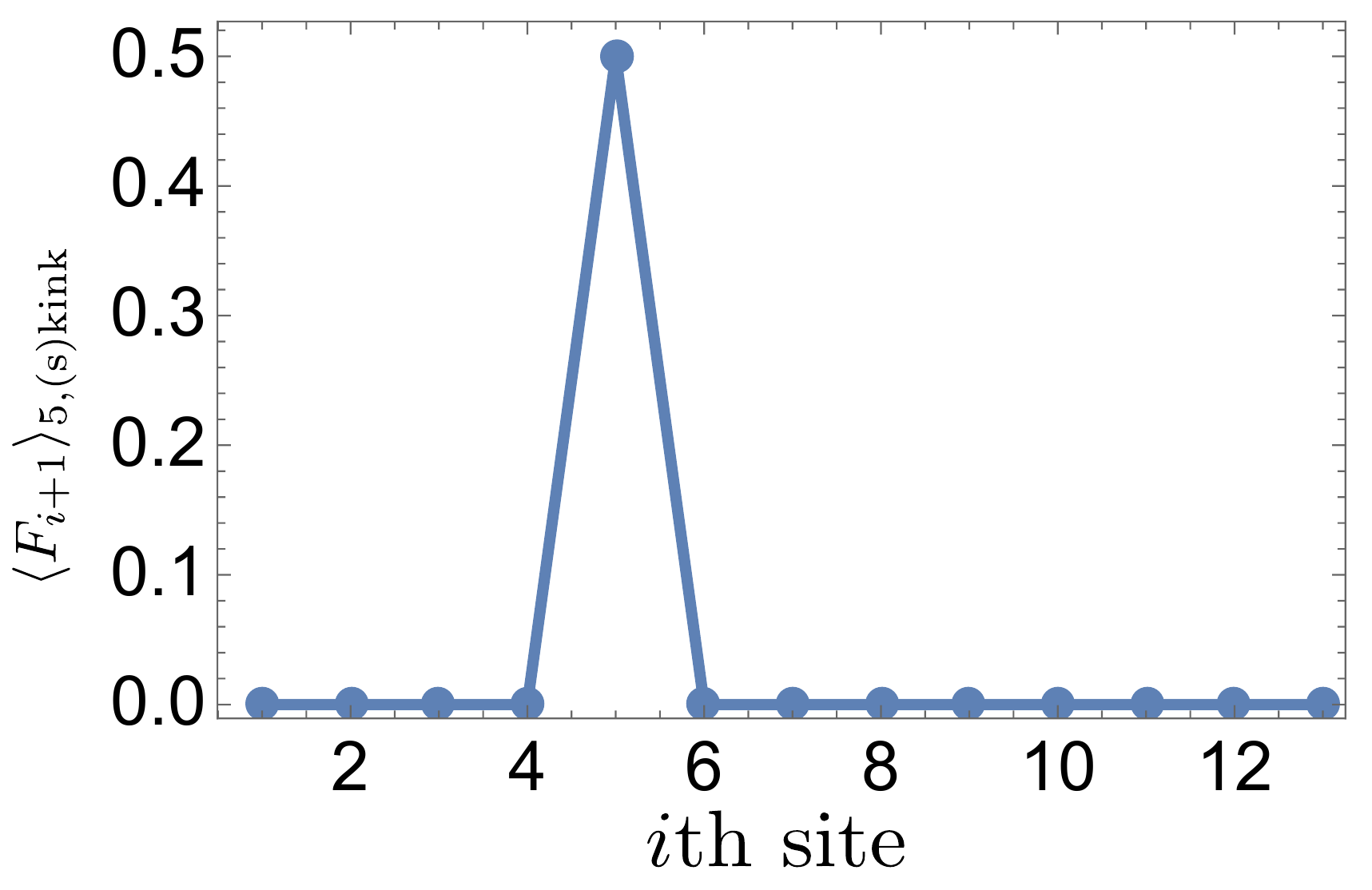}
\caption{The expectation value of the order parameter, $\kitaiti{F_{i+1}}_{j=5,\text{(s)kink}}:=\bra{5,\text{(s)kink}}F_{i+1} \ket{5,\text{(s)kink}}$, 
in a finite ($L=15$) open chain at $\alpha=\hf$.
The values at the boundaries ($i=0,14$) are undefined for OBC and are not included in the plot.  
In the region to the left of the kink (``trivial,'' i.e., ``quantum paramagnetic'') and that to the right (``topological,'' i.e., ``ferromagnetic''), SUSY is restored.  
SUSY is broken locally only at the location of the kink, which is detected by $\kitaiti{F_{i+1}}$ taking a nonzero value only at the kink's position.  
\label{Fig-fs}}
\end{figure}
%%%%%%%% FIG %%%%%%%%%%%%%%%%%%%%%%%%%%%%%%%%%%%%%%%%%%%

In certain supersymmetric lattice models, it is known that the number of the ground state grows exponentially with the system size, 
a phenomenon referred to as superfrustration \cite{huijse2008superfrustration,huijse2012supersymmetric,fendley2005exact,PhysRevLett.101.146406}.  
Although the large ground-state degeneracy at $\alpha=\hf$ found in the previous section appears somewhat similar to superfrustration, 
the degeneracy is not exponential in the system size, but rather of $O(L)$.  
Moreover, this behavior depends crucially on the boundary condition; under PBC, the ground-state degeneracy is always two [i.e., $O(1)$] 
except at $\alpha=1$.   
This suggests a qualitative difference between the large degeneracy due to a domain wall and superfrustration. 

Lastly, it is worth mentioning that for both periodic and open boundary conditions, the model is frustration free \cite{Garcia-V-W-C-07} 
when $\alpha=\hf$ in the sense that the ground state optimizes all the local Hamiltonians \eqref{eqn:local-Ham-by-spin} simultaneously 
(see, e.g., Refs.~\cite{tasaki2020physics,10.21468/SciPostPhysCore.4.4.027,o2018lattice} for other examples of frustration-free 
Hamiltonians) and that this enabled us to determine the ground states.

%%%%%%%%%%%%%%%%%%%%%%%%%%%%%%%%%%%%%%%%%%%%%%%%%%%
\subsection{Dispersion of kink and skink at $\alpha\neq\hf$}
%%%%%%%%%%%%%%%%%%%%%%%%%%%%%%%%%%%%%%%%%%%%%%%%%%%
In Sec.~\ref{skink}, we have found all the $2L$-degenerate ground states at $\alpha=\hf$ as the domain-wall (kink-skink) states 
which optimize the individual local Hamiltonians.  Being degenerate, these states represent an immobile domain wall.   
Then, one may ask if the domain wall acquires a finite dispersion when we move away from $\alpha=\hf$.  
In this section, we consider the effect of a small deviation $\Delta \alpha$ ($\alpha=\hf+\Delta \alpha$) on 
the domain-wall propagation using the first-order perturbation theory.   

We first expand the Hamiltonian $H^{\text{OBC}}\qty(\alpha=\hf+\Delta \alpha)$ up to first order in $\Delta \alpha$ as:
\begin{align}
H^{\text{OBC}}\qty(\alpha=\hf+\Delta \alpha)=H^{\text{OBC}}\qty(\hf)+ V \Delta \alpha +O\qty(\Delta \alpha^2) 
\end{align}
with $V$ being given by:
\begin{align}
V=H^{\text{OBC}}_{\text{int}}-\frac{L-2}{2}  \; .
\end{align} 
Since $V$ preserves the fermion parity, the matrix $V$ in the subspace spanned by the $2L$ kink-skink states in \eqref{eq:kink&skink} 
decomposes into two identical $L$-dimensional diagonal blocks corresponding to the kink and skink,  
whose matrix elements $\mathcal{V}_{i,j}\coloneqq\bra{i,\text{(s)kink}}V\ket{j,\text{(s)kink}}$ are given as follows: 
%%%%%%%%%%%%%%%%%%%%%%%%%%%%%%%%%%%%%%%
\begin{align}
\mathcal{V}_{i,j}=
\left\{
\begin{aligned}
&-\hf && \qty(i=j\leq L-2) \\
&\hf && \qty(i=j= L-1,L) \\
&0 && \qty(i=1 \text{ and } j\geq 2)\\
& && \text{ or } \qty(j=1 \text{ and } i\geq 2) \\
&-\qty(-\frac{1}{\sqrt{2}})^{\abs{i-j}-1} && \qty(i=L \text{ and } 2\leq j\leq L-1)\\
& &&\text{ or } \qty(j=L \text{ and } 2\leq i\leq L-1) \\
& \qty(-\frac{1}{\sqrt{2}})^{\abs{i-j}} && \qty(\text{otherwise})  \; .
\end{aligned}
\right.
\end{align}
%%%%%%%%%%%%%%%%%%%%%%%%%%%%%%%%%%%%%%%
For example, when $L=8$, the matrix $\mathcal{V}$ looks like:
\begin{equation}
\begin{split}
& \mathcal{V} = \\
&\begin{pmatrix}
-\hf & 0 & 0 & 0 & 0 & 0 & 0 & 0 \\
0 & -\hf & -\frac{1}{\sqrt{2}} & \frac{1}{2} & -\frac{1}{2\sqrt{2}} & \frac{1}{4} & -\frac{1}{4\sqrt{2}} & \frac{1}{4\sqrt{2}} \\
0 & -\frac{1}{\sqrt{2}} & -\hf & -\frac{1}{\sqrt{2}} & \frac{1}{2} & -\frac{1}{2\sqrt{2}} & \frac{1}{4} & -\frac{1}{4} \\
0 & \frac{1}{2} & -\frac{1}{\sqrt{2}} & -\hf & -\frac{1}{\sqrt{2}} & \frac{1}{2} & -\frac{1}{2\sqrt{2}} & \frac{1}{2\sqrt{2}} \\
0 & -\frac{1}{2\sqrt{2}} & \frac{1}{2} & -\frac{1}{\sqrt{2}} & -\hf & -\frac{1}{\sqrt{2}} & \frac{1}{2} & -\frac{1}{2} \\
0 & \frac{1}{4} & -\frac{1}{2\sqrt{2}} & \frac{1}{2} & -\frac{1}{\sqrt{2}} & -\hf & -\frac{1}{\sqrt{2}} & \frac{1}{\sqrt{2}} \\
0 & -\frac{1}{4\sqrt{2}} & \frac{1}{4} & -\frac{1}{2\sqrt{2}} & \frac{1}{2} & -\frac{1}{\sqrt{2}} & \hf & -1 \\
0 & \frac{1}{4\sqrt{2}} & -\frac{1}{4} & \frac{1}{2\sqrt{2}} & -\frac{1}{2} & \frac{1}{\sqrt{2}} & -1 & \hf \\
\end{pmatrix} 
\; .
\end{split}
\end{equation}

%%%%%%%%%%%%%%%%%%%%%%%%%%%%%%%%%%%%%%%%%%%%%%%%%%%%
We denote the eigenvalues of the matrix \( \mathcal{V} \) by \( \{ v_{n} \} \) (where \( n = 1, \ldots, L \)), with which the (s)kink energies are given as: 
\[ E_{n} = \frac{1}{8} + v_{n} \Delta \alpha + O(\Delta \alpha^2). \]
While we do not possess closed analytical expressions for \( \{ v_{n} \} \) for arbitrary system sizes \( L \), we have a good reason to believe that for sufficiently large \( L \), specifically when 
\[ L \gg -\frac{1}{\ln(\abs{-1/\sqrt{2}})} = 2.88539\ldots, \]
the eigenvalues \( \{ v_{n} \} \) can be well approximated by those of the following matrix:
\[ [\mathcal{V}^{\infty}]_{i,j} = \qty(-\frac{1}{\sqrt{2}})^{\abs{i-j}} - \frac{3}{2} \delta_{i,j} \quad (i,j \in \mathbb{N}), \]
obtained by retaining only the bulk part \( [\mathcal{V}]_{ij} \) (for \( 2 \leq i,j \leq L-1 \)), substituting \( [\mathcal{V}]_{L-1,L-1} = \frac{1}{2} \) with \( -\frac{1}{2} \), and then taking the limit \( L \to \infty \) (see Fig.~\ref{Fig-kinkspec} for the comparison of $v_{n}$ and the spectrum of $\mathcal{V}^{\infty}$).
%%%%%%%% FIG %%%%%%%%%%%%%%%%%%%%%%%%%%%%%%%%%%%%%%%%%%
\begin{figure}[htbp]
  \includegraphics[width=\columnwidth,clip]{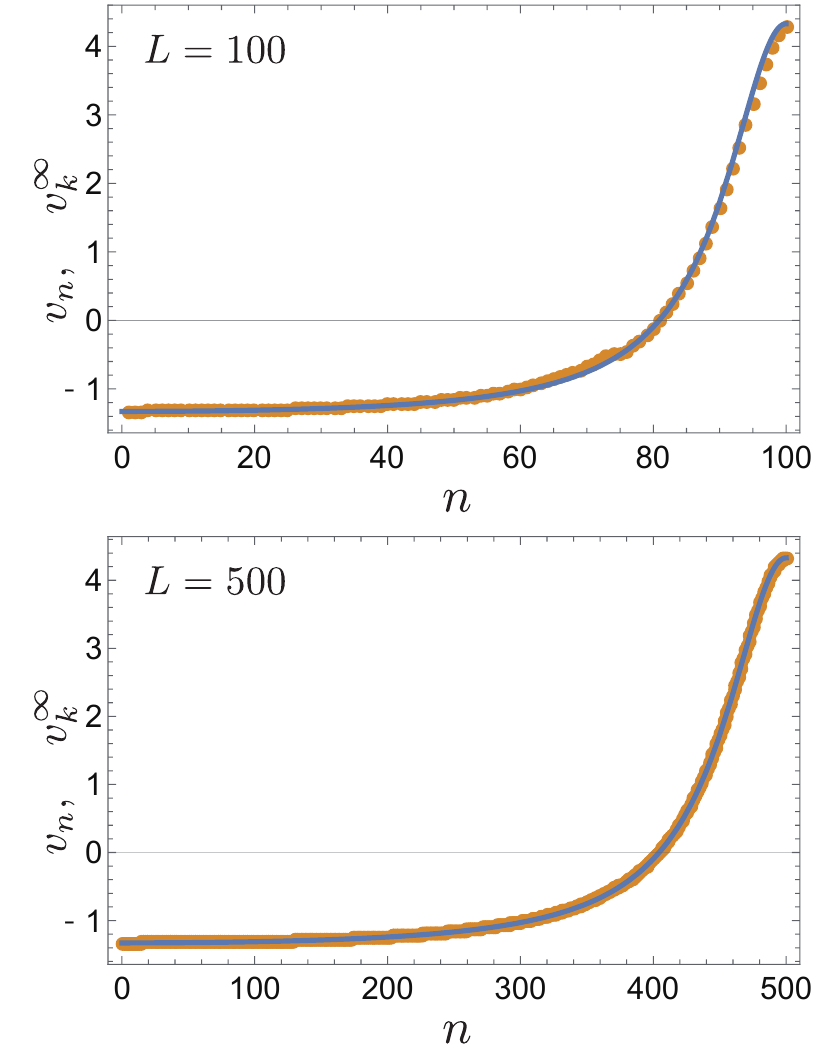}
  \caption{The orange dots displays the eigenvalues $v_{n}$ of $\mathcal{V}$ for system sizes $L=100$ (top) and $500$ (bottom).   
  The blue curve represents the values of $v_k^\infty$ (the data are joined).  
  While there is still a slight discrepancy between the two for $L=100$, they coincide almost perfectly for $L=500$.  
\label{Fig-kinkspec}}
\end{figure}
%%%%%%%% FIG %%%%%%%%%%%%%%%%%%%%%%%%%%%%%%%%%%%%%%%%%%%
Fourier-transforming \( \mathcal{V}^{\infty} \) and performing the summation of a geometric series yields:
\[ v_k^{\infty} = \frac{1}{3+2\sqrt{2}\cos(k)}-\frac{3}{2}, \]
where \( -\pi < k \leq \pi \).
%%%%%%%%%%%%%%%%%%%%%%%%%%%%%%%%%%%%%%%%%%%%%%%%%%%%

%%%%%%%% FIG %%%%%%%%%%%%%%%%%%%%%%%%%%%%%%%%%%%%%%%%%%
\begin{figure}[htbp]
  \includegraphics[width=\columnwidth,clip]{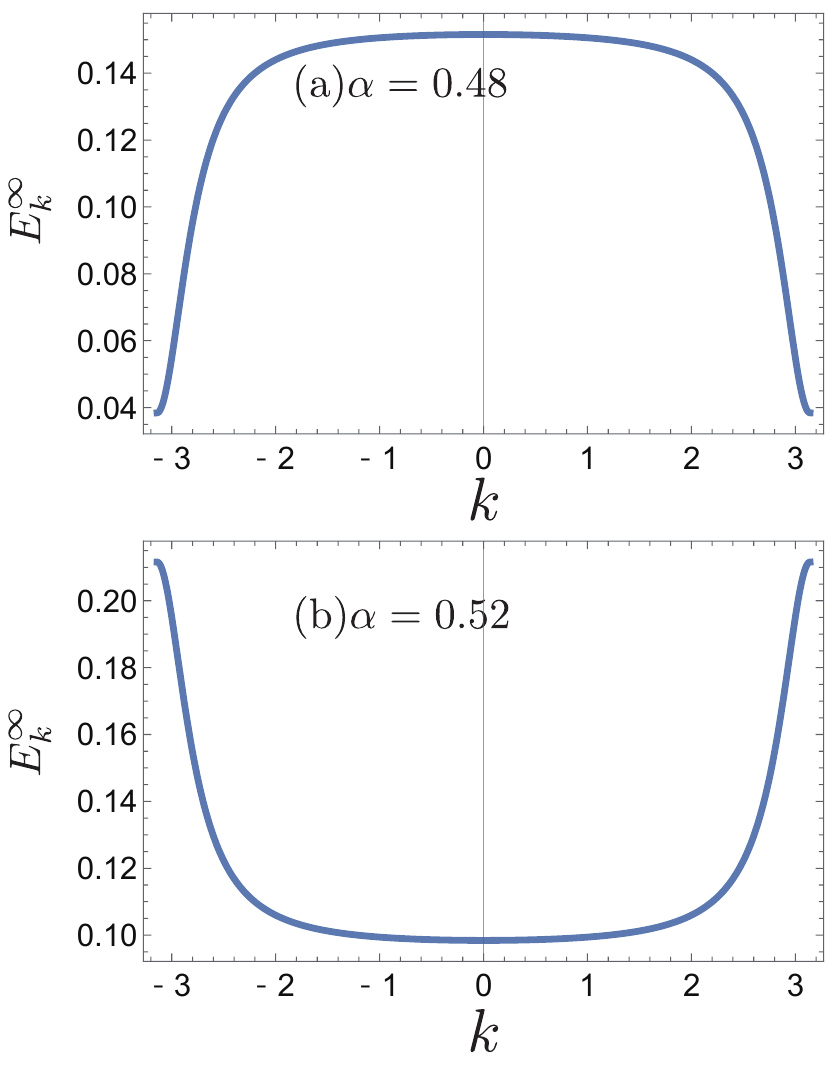}
  \caption{
 The eigenvalues $E_k$ are given by \eqref{eq:Ek}, where, in the thermodynamic limit, $E_k\coloneqq \frac{1}{8}+\Delta \alpha v_k^{\infty}+O\qty(\Delta \alpha^2)$. These eigenvalues are plotted at values of $\Delta\alpha$ being $\pm0.02$, corresponding to $\alpha=0.48$ and $0.52$. They exhibit symmetry in $k$ space, and around $k\sim 0$, the dispersion is approximately of the form $\text{const}+\text{const }k^2$. Therefore, based on the results of this perturbation theory, the system appears to be gapless around $\alpha\approx\hf$.
\label{Fig-kinkdispersion}}
\end{figure}
%%%%%%%% FIG %%%%%%%%%%%%%%%%%%%%%%%%%%%%%%%%%%%%%%%%%%%

Therefore, in the thermodynamic limit, the (s)kink spectrum of a propagating (s)kink
\begin{equation}
\ket{\Psi\qty(k),\text{(s)kink}}  \coloneqq \sum_{i=1}^{\infty} \eno{-\im k i}\ket{i,\text{(s)kink}},
\end{equation}
where the normalization is
\begin{equation}
\begin{split}
    &\bra{i,\text{(s)kink}}\ket{j,\text{(s)kink}}=\delta_{i,j}\\
    &\bra{\Psi\qty(k),\text{(s)kink}}\ket{\Psi\qty(p),\text{(s)kink}}=2\pi\delta\qty(k-p)\\
    &\bra{\Psi\qty(k),\text{(s)kink}}\ket{i,\text{(s)kink}}=\eno{\im ki},
\end{split}
\end{equation}
is given by:
\begin{equation}
\begin{split}
E_k&=\frac{1}{8}+ v_k^{\infty} \Delta \alpha +O\qty(\Delta \alpha^2)   \\
&=\frac{1}{8}+\Delta \alpha \qty(\frac{1}{3+2\sqrt{2}\cos(k)}-\frac{3}{2})+O\qty(\Delta \alpha^2) \; .
\end{split}
\label{eq:Ek}
\end{equation}
From this, we see that if we move away from $\alpha=\hf$, the $2L$-fold ground-state degeneracy is lifted down to two which is  the plane-wave states $\ket{k,\text{(s)kink}}$ with $k=0$ (when $\Delta \alpha >0$) 
or $\pi$ ($\Delta \alpha < 0$). 
Furthermore, the spectrum $E_k$, plotted in Fig.~\ref{Fig-kinkdispersion}, is consistent with the ``rounded'' plateau structure observed in Figs.~\ref{Fig-OBCgraphs} (b) and (d), suggesting the importance of a kink or skink in the low-energy physics. 

In the SUSY-broken phase $\alpha < \alpha_{\text{c}}$, especially for very small $\alpha\ll1$, 
the dominance of the $\alpha\qty(1-\alpha)H_{\text{CTI}}$ term and the mean-field approximation \eqref{eq:dispMF} suggest 
that the dispersion relation (under PBC) behaves like $\eps_k\sim k$ at small momenta.
For $\alpha\approx \hf$, on the other hand, by expanding the perturbative result \eqref{eq:Ek}, we see that the model shows a gapless quadratic spectrum 
$E_k- E_{\text{g.s.}} \sim k^2+O\qty(\Delta \alpha^2)$ over the ground states.   
When $\alpha=1$, as reported in \cite{fendley2019free}, the single-particle spectrum exhibits $\eps_k\sim \abs{\pi-k}^{\frac{3}{2}}$. 
Combining the results obtained in Sec.~\ref{sec:NG-fermion-in-MFA} and here, we conclude 
that the system remains gapless at least for $\alpha\approx \hf$ and $\alpha=1$, as well as for $\alpha\leq \alpha_{\text{c}}$ (due to the NG fermion).

%%%%%%%%%%%%%%%%%%%%%%%%%%%%%%%%%%%%%%%%%%%%%%%%%%%%%%%%%%%%
\section{Conclusion}
\label{sec:conclusion}
%%%%%%%%%%%%%%%%%%%%%%%%%%%%%%%%%%%%%%%%%%%%%%%%%%%%%%%%%%%%
In conclusion, we have introduced and thoroughly investigated a one-dimensional model of interacting Majorana fermions with exact $\mathcal{N}=1$ supersymmetry that includes the Kitaev chain and the model studied in Ref.~\cite{fendley2019free} as the noninteracting and interaction parts. By applying the Jordan-Wigner transformation to this Majorana chain, it can be represented as a transverse-field Ising model where three spins interact. When the control parameter $\alpha$ is increased, the system undergoes a phase transition from the phase with broken SUSY to the SUSY-symmetric one. By combining a variational argument and the exact ground state obtained at the frustration-free point $\alpha=\hf$, we have concluded that the transition occurs in the interval $\frac{\pi}{8} < \alpha_\text{c} < \hf$. 

We have then developed a mean-field theory based on the superfield formalism. The theory contains a pair of order parameters $F$ and $G$ which detect SUSY-breaking. The behavior of the ground-state energy density (which is also an order parameter for SUSY SSB) and the phase structure predicted by the theory are consistent with those expected from numerical analyses. Within the mean-field theory, it is straightforward to obtain the spectrum of the NG fermion which is anticipated in the SUSY-broken phase.  We have found that it is gapless $k$-linear and that the characteristic velocity vanishes at the boundaries of the SUSY-broken phase. Combining this with the results of numerical (Sec.~\ref{sec:numerical-spec}) and analytical (Appendix~\ref{App:NG-single}) calculations, we expect that a critical phase of the Ising universality class ($c=\hf$) extends over the SUSY-broken region $0 \leq \alpha < \alpha_{\text{c}}$, in which the massless Majorana fermion plays the role of the NG fermion. 

The property of the model becoming frustration free at $\alpha=\hf$ is common to both PBC and OBC and allows us to obtain interesting exact results. For PBC, we have shown that the topological and trivial phases of the Kitaev chain degenerate to form a pair of zero-energy ground states. On the other hand, for OBC, the system allows an immobile zero-energy domain wall ({\em kink} or its superpartner {\em skink}) between the above topological and trivial phases to appear in the ground state, thereby leading to degeneracy proportional to the system size. 

To see how this ground-state degeneracy is resolved when we move away from the frustration-free point, 
we have carried out first-order perturbation theory in $\alpha\sim\hf$. The behavior of the degenerate ground states in the thermodynamic limit has been precisely determined which suggests that the degeneracy is resolved away from the frustration point leaving only the twofold one required by supersymmetry. As seen in Fig.~\ref{Fig-kinkdispersion}, when $\alpha \neq \hf$, the $2L$ degenerate ground states form a gapless spectrum (behaving like $\sim k^2$) over the (doubly degenerate) ground state. Our perturbative calculations based on the kink-skink picture well reproduce the numerically observed spectra suggesting the importance of these domain-wall excitations in understanding the low-energy properties of the supersymmetric model. 

Through analytical and numerical analyses presented here, the properties in the region \( \alpha \leq \alpha_\text{c} \) have been relatively well understood. However, whereas the region $\hf<\alpha < 1$ remains largely unexplored. In particular, little is known about the low-lying excitation spectrum in that region, which is crucial to a full understanding of the crossover from the kink-skink physics around $\alpha \approx \hf$ to the superfrustration behavior at $\alpha=1$. Investigating low-energy properties with extensive numerical simulations would be an interesting subject for future projects. 
Yet another essential open problem is uncovering the nature of the SUSY-restoring transition at  $\alpha_{\text{c}}$. In this respect, combining large-scale numerical simulations and the techniques of conformal field theories will be very useful.

%%%%%%%%%%%%%%%%%%%%%%%%%%%%%%%%%%%%%%%%%%%%%%%%%%%%%%%%%%%%%%
\section*{Acknowledgments}
%%%%%%%%%%%%%%%%%%%%%%%%%%%%%%%%%%%%%%%%%%%%%%%%%%%%%%%%%%%%%%
The authors would like to thank T.~Ando, A.~Ueda, and Y.~Nakayama, for helpful discussions, and especially H.~Katsura for discussions and 
letting us know about a paper \cite{sannomiyaDron} where some of the results 
in Secs.~\ref{sec:Model} and \ref{sec:SUPERSYMMETRY BREAKING} have been obtained independently. 
The author (K.T.) is supported in part by Japan Society for the Promotion of Science (JSPS) KAKENHI Grant No. 21K03401 and the IRP project ``Exotic Quantum Matter in Multicomponent Systems (EXQMS)'' from 
CNRS.
The author (K.S.) is supported by Graduate School of Science, Kyoto University under Ginpu Fund.
%%%%%%%%%%%%%%%%%%%%%%%%%%%%%%%%%%%%%%%%%%%%%%%%%%%%%%%%%%%%%%%%%%%%%%%%%
\appendix
%%%%%%%%%%%%%%%%%%%%%%%%%%%%%%%%%%%%%%%%%%%%%%%%%%%%%%%%%%%%%%%%%%%%%%%%%
\section{Supercharge in the Jordan-Wigner transformed spin language}
\label{app:spin-supercharge}
In Sec.~II A, we have given the supercharge $R\qty(\alpha)$ in terms of the Majorana fermions.
By the Jordan-Wigner transformation \eqref{eqn:JW-tr}, the fermionic expressions \eqref{eqn:superchargePBC} and \eqref{eqn:superchargeOBC} are mapped to the following representations that are highly nonlocal in terms of spins:
\begin{equation}
\begin{split}
& R^{\text{PBC}}\qty(\alpha) \\
&=\qty(1-\alpha)\sum_{i=1}^{L}\sx_i\prod_{j\qty(<i)}\qty(-\sz_j)
    -\alpha\sum_{i=1}^{L-1}\sx_{i+1}\prod_{j\qty(<i)}\qty(-\sz_j)  \\
&    \phantom{=}  +\alpha\sz_{L}\sx_{1}
\intertext{and}
& R^{\text{OBC}}\qty(\alpha)    \\
&= \qty(1-\alpha)\sum_{i=2}^{L-1}\sx_i\prod_{j\qty(<i)}\qty(-\sz_j)
- \alpha\sum_{i=1}^{L-1}\sx_{i+1}\prod_{j\qty(<i)}\qty(-\sz_j)  \; ,
\end{split}
\end{equation}
respectively.

%%%%%%%%%%%%%%%%%%%%%%%%%%%%%%%%%%%%%%%%%%%%%%%%%%%%%%%%%%%%%%%%%%%%%%%%%
%%%%%%%%%%%%%%%%%%%%%%%%%%%%%%%%%%%%%%%%%%%%%%%%%%%%%%%%%%%%%%%%%%%%%%%%%
\section{Diagonalization of $H_{\rmm{CTI}}$ 
and its ground-state energy}
\label{App:free}
%%%%%%%%%%%%%%%%%%%%%%%%%%%%%%%%%%%%%%%%%%%%%%%%%%%%%%%%%%%%%%%%%%%%%%%%%
In the small \( \alpha \) limit, the first significant term arises from the free part \( \alpha(1-\alpha)H_{\text{CTI}} \). This appendix is devoted to diagonalizing this free part of the Hamiltonian.

To achieve this, we employ the discrete Fourier transform defined as:
\begin{align}
  \beta_j &= \sqrt{\frac{2}{L}}\sum_{k>0}\left(\eno{-\im k j}b_k+\eno{\im k j}b_k^\dagger\right) \\
  \gamma_j &= \sqrt{\frac{2}{L}}\sum_{k>0}\left(\eno{-\im k j}c_k+\eno{\im k j}c_k^\dagger\right) .
\end{align}
With this transformation, the \( H_{\text{CTI}} \) can be expressed in the following form:
\begin{align}
  H_{\text{CTI}}
  &= 2\sum_{k>0} 
  \begin{pmatrix}
      b_k^\dagger &
      c_k^\dagger
  \end{pmatrix}
  \begin{pmatrix}
      0 & \im \qty(\eno{\im k}-1) \\
      -\im \qty(\eno{-\im k}-1) & 0
  \end{pmatrix}
  \begin{pmatrix}
      b_k \\ 
      c_k
  \end{pmatrix} \\
  &= \sum_{k} 4\abs{\sin(\frac{k}{2})} \left(f_k^\dagger f_k - \frac{1}{2}\right).
\end{align}
Here, \( f_k \) represents the quasiparticle fermion operator.
The ground state of \( H_{\text{CTI}} \) corresponds to the vacuum state of the quasiparticle \( f_k \). The energy associated with this state is represented by the zero-point energy, which can be determined as:
\begin{align}
\label{eq:Egs-density}
  e_0^{\text{CTI}} = -\int_{-\pi}^{\pi} \frac{dk}{2\pi} 2\abs{\sin(\frac{k}{2})} = -\frac{4}{\pi}.
\end{align}

%%%%%%%%%%%%%%%%%%%%%%%%%%%%%%%%%%%%%%%%%%%%%%%%%%%%%%%%%%%%%%%%%%%%%%%%%
\section{Diagonalization of $H_{\rmm{mf}}$ and its ground-state energy}
\label{App:mf}
%%%%%%%%%%%%%%%%%%%%%%%%%%%%%%%%%%%%%%%%%%%%%%%%%%%%%%%%%%%%%%%%%%%%%%%%%
To find the ground-state energy using the mean-field approximation and discuss SUSY breaking, we need to diagonalize the mean-field Hamiltonian \eqref{eq:mfHam}.
%%%%%%%%%%%%%%%
\begin{widetext}
By applying the discrete Fourier transform 
\begin{equation}
\beta_j=\sqrt{\frac{2}{L}}\sum_{k>0}\qty(\eno{-\im k j}b_k+\eno{\im k j}\hc{b}_k) \; , \quad 
\gamma_j=\sqrt{\frac{2}{L}}\sum_{k>0}\qty(\eno{-\im k j}c_k+\eno{\im k j}\hc{c}_k),
\end{equation}
we can rewrite $H_{\rmm{mf}}$ as:
\begin{equation}
\label{eq:dipersionMF}
\begin{split}
H_{\rmm{mf}} =&  2\alpha\sum_{k>0} 
  \begin{pmatrix} 
    \hc{b}_k & \hc{c}_k 
  \end{pmatrix}
  \begin{pmatrix}
      2G\sin(k) & \im F(\eno{\im k}-1) \\
      -\im F(\eno{-\im k}-1) & 0
  \end{pmatrix}
  \begin{pmatrix}
      b_k \\ 
      c_k
  \end{pmatrix} 
  -2\alpha G\sum_{k>0}\sin(k) + \qty(\qty(1-\alpha)F-\frac{1}{2}F^2-\frac{1}{2}G^2)L   \\
 =&  \sum_{k}  \qty(4\alpha\abs{\sin(\frac{k}{2})}\sqrt{F^2+G^2\cos^2\qty(\frac{k}{2})}
  +2\alpha G\sin(k))
  \hc{f}_k f_k 
  -\sum_{k}2\alpha\abs{\sin(\frac{k}{2})}
  \sqrt{F^2+G^2\cos^2\qty(\frac{k}{2})}\nt
&  +\qty(\qty(1-\alpha)F-\hf F^2-\hf G^2)L  \; , 
\end{split}
\end{equation}
where $f_k$ represents the quasiparticle fermion operator.

The ground state of $H_{\rmm{mf}}$ is the vacuum of the quasiparticle $f_k$, its energy is given by  
the zero-point energy of 
\begin{align}
  \label{eq:mean-field-Edensity}
  e^{\rmm{mf}}_0\qty(\alpha;F,G)=&-\int^{\pi}_{-\pi}\frac{\rmm{d}k}{2\pi}
  2\alpha\abs{\sin(\frac{k}{2})}
  \sqrt{F^2+G^2\cos^2\qty(\frac{k}{2})} +\qty(1-\alpha)F-\hf F^2-\hf G^2 \nt
  =&\qty(1-\alpha)F-\hf F^2-\hf G^2 -\frac{2\alpha}{\pi}\qty(
  \sqrt{F^2+G^2}+\frac{F^2}{\abs{G}}
  \sinh^{-1}\qty(\abs{\frac{G}{F}})
  ).
\end{align}

\end{widetext}
%%%%%%%%%%%%%%%%%%%%%%%%%

%%%%%%%%%%%%%%%%%%%%%%%%%%%%%%%%%%%%%%%%%%%%%%%%%%%%%%%%%%%%%%%%%%%%%%%%%
\section{Order parameters in the infinite size systems}
\label{App:orderineq}
%%%%%%%%%%%%%%%%%%%%%%%%%%%%%%%%%%%%%%%%%%%%%%%%%%%%%%%%%%%%%%%%%%%%%%%%%
In this section, in an infinite system, we consider $\ket{0}$ as one of the ground states of the system, and define the space-averaged order parameters as follows:

\begin{align}
\kitaiti{F}:&=\frac{1}{L}\sum_{j=1}^L\bra{0}F_j\ket{0} \nt
\kitaiti{G}:&=\frac{1}{L}\sum_{j=1}^L\bra{0}G_j\ket{0}.
\end{align}

These order parameters play a role in detecting the spontaneous breaking of SUSY ($e_o^{\rmm{gs}}>0$) in the infinite system.
The calculations and logic used here are based on \cite{moriya2018supersymmetry}.

Let us introduce $x_j$ as a fermionic local operator that satisfies the following conditions:

\begin{align}
  \label{eq:conditionX}
  &\acom{\qty(-1)^F}{x_j}=0 \nt
  &\hc{x}_j=x_j \nt
  &\bra{0} \acom{x_i}{x_j} \ket{0}=
  \left\{
\begin{aligned}
&O\qty(1) && \qty(\abs{i-j}\leq r) \\
&0 && \qty(\abs{i-j}>r).
\end{aligned}
\right.
\end{align}

Here, \( x_j \) denotes a fermionic local operator, while \( X \) represents its spatial average:

\begin{align}
X\coloneqq\frac{1}{L}\sum_{j=1}^L x_j,
\end{align}
where the parameter $r$ represents the interaction range ($O(1)$ natural number). The Hermitian condition is imposed for simplicity, but the same discussion can be done without it (the $x_j$ in the main text, corresponding to $\beta_j$ and $\gamma_j$, are Hermitian).

Given this, we define \( \Phi_j \) and \( \kitaiti{\Phi} \) as:
\begin{align}
&\Phi_j\coloneqq\acom{R}{x_j} \nt
&\kitaiti{\Phi}\coloneqq\frac{1}{L}\sum_{j=1}^L\bra{0}\Phi_j\ket{0}.
\end{align}
Then, we have:
\begin{align}
&\kitaiti{\Phi}=\bra{0}\acom{R}{X}\ket{0}.
\end{align}
Similar to the case of \( F_j \) and \( G_j \) in Section \ref{sec:maenfield}, these quantities serve as order parameters in finite systems.
The main question is whether the detection of SUSY SSB $e_o^{\rmm{gs}}>0$ remains valid in the infinite system.

Assuming SUSY unbroken, we have $\hc{R}=R$, and thus:

\begin{align}
0&\leftarrow e_o^{\rmm{gs}} =\frac{\bra{0}H\ket{0}}{L}=\frac{\bra{0}R^2\ket{0}}{2L}
=\frac{\norm{R\ket{0}}^2}{2L}.
\end{align}
This implies:
\begin{align}
\label{eq:normR}
\norm{R\ket{0}}=o\qty(\sqrt{L}).
\end{align}

From \eqref{eq:conditionX}, we find:
\begin{align}
\norm{X\ket{0}}^2&=\bra{0}X^2\ket{0}\nt
&=\frac{1}{2L^2}\sum_{i,j}\bra{0} \acom{x_i}{x_j} \ket{0}\nt
&=\frac{1}{2L^2}O\qty(L) =O\qty(1/L).
\end{align}
This leads to:
\begin{align}
\label{eq:normX}
\norm{X\ket{0}}=O\qty(1/\sqrt{L})
\end{align}

Now, using the triangle inequality, the Cauchy-Schwarz inequality, and the Hermiticity of $R$ and $X$, we find:
\begin{align}
\abs{\kitaiti{\Phi}}&=\abs{\bra{0}\qty(RX+XR)\ket{0}} \nt
&\leq2\norm{R\ket{0}}\norm{X\ket{0}} \nt
&=2o\qty(\sqrt{L})O\qty(1/\sqrt{L}) \nt
&=o\qty(1).
\end{align}
Thus, if we take the contrapositive, we can say that if $\abs{\kitaiti{\Phi}}\geq O\qty(1)$, then SUSY SSB is occurring in the infinite system. Therefore, the order parameter constructed in this way, which represents the spatial average of $\Phi_j=\acom{R}{x_j}$, can indeed detect SUSY SSB in the infinite system when it becomes nonzero.

%%%%%%%%%%%%%%%%%%%%%%%%%%%%%%%%%%%%%%%%%%%%%%%%%%%%%%%%%%%%%%%%%%%%%%%%%
\section{Expectation value of order parameters}
\label{App:orderparameter}
%%%%%%%%%%%%%%%%%%%%%%%%%%%%%%%%%%%%%%%%%%%%%%%%%%%%%%%%%%%%%%%%%%%%%%%%%
In this section, we demonstrate that the expectation value of the order parameter remains the same for the doublet with energy $E>0$ under supersymmetry. We adopt the notation from Appendix \ref{App:orderineq}.

For the doublet, we denote the states with fermion parity $+1$ as $\ket{E,+}$ and those with fermion parity $-1$ as $\ket{E,-}$. Due to the properties of the operator \( R \), the following relations hold:

\begin{align}
&R\ket{E,+}=\sqrt{2E}\ket{E,-} \nt
&R\ket{E,-}=\sqrt{2E}\ket{E,+}.
\end{align}

Since $\hc{R}=R$, we have:

\begin{align}
\bra{E,+}\Phi\ket{E,+}&=\frac{1}{2E}\bra{E,-}\hc{R}\Phi R\ket{E,-} \nt
&=\frac{1}{2E}\bra{E,-}R \acom{R}{X} R\ket{E,-} \nt
&=\frac{1}{2E}\bra{E,-}\qty(2HXR+RX2H) \ket{E,-} \nt
&=\bra{E,-} \acom{X}{R} \ket{E,-} \nt
&=\bra{E,-}\Phi\ket{E,-}.
\end{align}

Thus, we have demonstrated that the expectation value of the order parameter is the same for both $\ket{E,+}$ and $\ket{E,-}$ states. Therefore, when determining the value of the order parameter, one can choose either state for calculation as long as proper normalization is performed.

%%%%%%%%%%%%%%%%%%%%%%%%%%%%%%%%%%%%%%%%%%%%%%%%%%%%%%%%%%%%%%%%%%%%%%%%%
\section{Ground state energy variation with boundary conditions}
\label{App:specvary}
%%%%%%%%%%%%%%%%%%%%%%%%%%%%%%%%%%%%%%%%%%%%%%%%%%%%%%%%%%%%%%%%%%%%%%%%%
In this section, we prove that when the difference in the Hamiltonian due to different boundary conditions is of order $O(1)$ with respect to the system size $L$, so is the difference in the ground-state energy. This result allows us to determine the existence/absence of SUSY SSB in the infinite system without being affected by the boundary conditions.

Consider \( H_1 \) and \( H_2 \) as the Hamiltonians corresponding to boundary conditions 1 and 2, respectively. 
They differ only by at most $O(1)$ terms associated to, e.g., those near the boundary. Therefore, in terms of the operator norm, we have
\begin{align}
  \label{eq:opsub}
  \norm{H_2 - H_1} = O(1).
\end{align}
We assume that each local Hamiltonian has an operator norm of order \( O(1) \).

Let $\ket{1}$ and $\ket{2}$ be the ground states of $H_1$ and $H_2$, respectively, and let $E_1$ and $E_2$ be the corresponding energies. Without loss of generality, we may assume $E_2 \geq E_1$. We want to show that the difference $E_2 - E_1$ is of order $O(1)$. 
By the variational principle and the definition of the operator norm, we have
\begin{align}
  0 \leq E_2 - E_1 &= \bra{2} H_2 \ket{2} - \bra{1} H_1 \ket{1} \nt
  &\leq \bra{1} H_2 \ket{1} - \bra{1} H_1 \ket{1} \nt
  &= \bra{1} (H_2 - H_1) \ket{1} \nt
  &\leq \norm{H_2 - H_1},
\end{align}
which, combined with Eq.~\eqref{eq:opsub}, implies \( E_2 - E_1 = O(1) \).
Therefore, we can conclude that the criterion \eqref{eqn:SSB-criterion} of SUSY SSB in the infinite system leads to the same conclusion regardless of the boundary conditions.

%%%%%%%%%%%%%%%%%%%%%%%%%%%%%%%%%%%%%%%%%%%%%%%%%%%%%%%%%%%%%%%%%%%%%%%%%
\section{Exact excited states at $\alpha=\hf$}
\label{App:exact-excited}
%%%%%%%%%%%%%%%%%%%%%%%%%%%%%%%%%%%%%%%%%%%%%%%%%%%%%%%%%%%%%%%%%%%%%%%%%

On top of the exact ground states discussed in Sec.~\ref{skink}, we can even find some exact excited states at $\alpha=\hf$ regardless of the system size $L$ when OBC is imposed.  
Hereinafter, we use the spin representation obtained via the Jordan-Wigner transformation.  
Then, it is straightforward to see that the following set of four states with the fermion parity $(-1)^F=+1$,
\begin{align}
&\ket{\dw\cdots\dw\up\up\dw\dw},\quad
\ket{\dw\cdots\dw\up\dw\up\dw},\quad
\ket{\dw\cdots\dw\up\dw\dw\up},\quad
\ket{\dw\cdots\dw\up\up\up\up},
\label{eq:exact-excited-even}
\end{align}
and its superpartner consisting of
\begin{align}
&\ket{\dw\cdots\dw\up\dw\dw\dw},\quad
\ket{\dw\cdots\dw\up\up\up\dw},\quad
\ket{\dw\cdots\dw\up\up\dw\up},\quad
\ket{\dw\cdots\dw\up\dw\up\up}
\label{eq:exact-excited-odd}
\end{align}
with $(-1)^F=-1$ are both closed under the action of the Hamiltonian $H(\alpha=\hf)$ [precisely speaking, for these sets to be correctly mapped to each other by $R_{\text{OBC}}(\alpha=\hf)$, appropriate linear combinations must be made].  
By diagonalizing the four-dimensional block Hamiltonian, we obtain the following exact energies:
\begin{align}
E=\frac{1}{8},\quad \frac{9-4\sqrt{2}}{8},\quad \frac{9}{8},\quad \frac{9+4\sqrt{2}}{8},
\label{eq:exact-excited-energies}
\end{align}
the latter three of which correspond to excited states.  
Comparing these three values with the energy levels obtained numerically [see Fig.~\ref{Fig-OBCgraphs}(c)], we see that they represent higher-lying excited states.

%%%%%%%%%%%%%%%%%%%%%%%%%%%%%%%%%%%%%%%%%%%%%%%%%%%%%%%%%%%%%%%%%%%%%%%%%
\section{Single-mode approximation for NG fermion}
\label{App:NG-single}
%%%%%%%%%%%%%%%%%%%%%%%%%%%%%%%%%%%%%%%%%%%%%%%%%%%%%%%%%%%%%%%%%%%%%%%%%

Using a variational argument, we can establish an upper bound for the dispersion of the NG fermions, thus confirming the existence of gapless excitations. In the realm of conventional broken symmetries, the variational state for the corresponding NG bosons is generated by applying the operators associated with the broken symmetries to the ground state (referred to as the single-mode approximation \cite{PhysRev.94.262,horsch1988spin,PhysRevB.49.6710,momoi1994upper}). Similarly, we can construct a trial state for the NG fermions by substituting the bosonic generators with the fermionic supercharge.
Specifically, following Refs.~\cite{sannomiya2016supersymmetry,sannomiya2017supersymmetry,sannomiya2019supersymmetry}, 
we consider the local supercharge operator
\begin{align}
r_i \coloneqq \frac{1-\alpha}{2}\qty(\beta_i+\beta_{i+1})+\alpha\im \beta_i\beta_{i+1}\gamma_i 
\end{align}
and introduce its momentum-$p$ component as:
\begin{align}
R_p \coloneqq \sum_{j =1}^L \eno{\im p j}r_j =\hc{R}_{-p}.
\end{align}
The trial variational state with momentum \( p \) is then defined as:
\begin{align}
\ket{\psi_p} \coloneqq \frac{R_p\ket{0}}{\norm{R_p\ket{0}}}\quad\qty(p\neq 0),
\end{align}
where $\ket{0}$ is the normalized ground state.  
The variational energy for this trial state \( \ket{\psi_p} \) is given
with the help of the same inequalities as 
in Refs.~\cite{sannomiya2016supersymmetry,sannomiya2017supersymmetry,sannomiya2019supersymmetry}.
Here, using the identity that follows from the inversion symmetry $\mathcal{I}$ or the time-reversal symmetry $\mathscr{T}$ in \ref{subSec:symmetry} possessed by the system:
\begin{align}
    &\mathcal{I}R_p\hc{\mathcal{I}}=\eno{-\im p}\hc{R}_p \nt
    &\mathscr{T}R_p\hc{\mathscr{T}}=\hc{R}_p \\
    &\bra{0}R_p\hc{R}_p\ket{0}=\bra{0}\hc{R}_pR_p\ket{0} \nt
    &\bra{0}R_p H \hc{R}_p\ket{0}=\bra{0}\hc{R}_p H R_p\ket{0},
\end{align}
we obtain:
\begin{align}
\eps_{\rmm{var}}\qty(p) :&= \bra{\psi_p}H\ket{\psi_p}-E_{\text{g.s.}} \nt
&= \frac{\bra{0}\com{\hc{R}_p}{{\com{H}{R_p}}}\ket{0}}
{\bra{0}\acom{\hc{R}_p}{R_p}\ket{0}} \nt
&\leq \sqrt{\frac{\bra{0}\acom{\com{\hc{R}_p}{H}}{\com{H}{R_p}}\ket{0}}
{\bra{0}\acom{\hc{R}_p}{R_p}\ket{0}}} \nt
&= \hf\sqrt{\frac{C}{E_{\text{g.s.}}/L}}\abs{p}+O\qty(\abs{p}^3),
\end{align}
where $C$ is a positive constant independent of $L$, and $E_{\text{g.s.}}/L \, (\gneq 0)$ is the ground-state-energy density.
Using a supersymmetric current \( j_j \), which satisfies the following relation
\begin{align}
    \com{\im H}{r_j}=j_{j}-j_{j+1},
\end{align}
the value of \( C \) can be expressed:
\begin{align}
    C=\frac{1}{L}\sum_{i,j}\bra{0}\acom{j_i}{j_j}\ket{0}
\end{align}
Here, the expression for the current is given by
\begin{align}
    &j_j \nonumber \\
    &=\alpha\qty(1-\alpha)^2\qty(\gamma_{j-1}+\gamma_{j}) \nonumber \\
    &-\alpha^2\qty(1-\alpha)\left(
    \im \beta_{j-2}\gamma_{j-2}\gamma_{j-1}
    +2\im \beta_{j-1}\gamma_{j-1}\gamma_{j} \right.\nonumber \\
    &\left.+2\im \beta_{j+1}\gamma_{j-1}\gamma_{j}+\im \beta_{j+2}\gamma_{j}\gamma_{j+1}
    \right) \nonumber \\
    &+2\alpha^3\qty(
    \beta_{j-2}\beta_{j+1}\gamma_{j-2}\gamma_{j-1}\gamma_{j}
    +\beta_{j-1}\beta_{j+2}\gamma_{j-1}\gamma_{j}\gamma_{j+1}
    ) \; .
\end{align}
This variational energy provides an upper bound on the true dispersion of the Nambu-Goldstone mode.
This fact can be obtained straightforwardly by extending the sum rule in bosonic systems \cite{PhysRev.94.262,horsch1988spin,PhysRevB.49.6710,momoi1994upper} to fermionic systems.
Specifically, the spectrum of the NG fermions is bounded from above by a $p$-linear dispersion.
Here, it should be noted that this variational approach is not effective when the trial state $\ket{\psi_p}$ coincides with other orthogonal ground states $\ket{0}$.
In periodic boundary conditions (PBC), we have numerically confirmed that the ground-state degeneracy is 2, except for the cases when $\alpha=0,1$. In these cases, there is no concern about the trial state being identical to one of the ground states.

%%%%%%%%%%%%%%%%%%%%%%%%%%%%%%%%%%%%%%%%%%%%%%%%%%%%%%%%%%%%%%%%%%%%%%
%\section{Ground states for \texorpdfstring{$\alpha=\hf$}{TEXT}}
\section{Ground states for $\alpha=\hf$}
\label{sec:proof-2-GS}
%%%%%%%%%%%%%%%%%%%%%%%%%%%%%%%%%%%%%%%%%%%%%%%%%%%%%%%%%%%%%%%%%%%%%%
In this section, we prove the following two theorems about ground states of the Hamiltonian $H\qty(\alpha=\hf)$.
\begin{thm}\label{thm:PBC}
    The ground-state energy of the Hamiltonian $H\qty(\alpha=\hf)$ under the periodic boundary condition is zero, 
    and the eigenspace $\mathcal{E}_{\mathrm{PBC}}$ spanned by the ground states is of the form
    \begin{align}\label{eq:thmPBC}
        \mathcal{E}_{\mathrm{PBC}}
        \coloneqq \ker H\qty(\alpha=\hf)
        = \mathbb{C}\ket{\mathrm{A}}+\mathbb{C}\ket{\mathrm{B}},
    \end{align}
    where 
   \begin{align}
        \ket{\mathrm{A}}\coloneqq \frac{\ket{\rightarrow}^{\otimes L}-(-1)^L\ket{\leftarrow}^{\otimes L}}{\sqrt{2}},
        \quad
        \ket{\mathrm{B}}\coloneqq\ket{\downarrow}^{\otimes L}.
    \end{align}
    In particular, $\dim\mathcal{E}_{\mathrm{PBC}}=2$.
\end{thm}
\begin{thm}\label{thm:OBC}
    The ground-state energy of the Hamiltonian $H\qty(\alpha=\hf)$ under the open boundary condition in \eqref{eqn:H^OBC} is $\frac{1}{8}$, 
    and the eigenspace $\mathcal{E}_{\mathrm{OBC}}$ spanned by the ground states is of the form 
    \begin{align*}
        &\mathcal{E}_{\mathrm{OBC}}
        \coloneqq \ker\left( H\qty(\alpha=\hf) -\frac{1}{8}\right) \\
        &= \sum_{j=1}^{L-2}\left(\mathbb{C}\ket{j,\rightarrow}_L+\mathbb{C}\ket{j,\leftarrow}_L\right)
        + \ket{\downarrow}^{\otimes (L-2)}\otimes(\mathbb{C}^2)^{\otimes 2}.
    \end{align*}
    In particular, $\dim\mathcal{E}_{\mathrm{OBC}}=2L$.
\end{thm}
%%%%%%%%%%%%%%%%%%%%%%%%%%%%%%%%%%%%%%%%%%%%%%%%%%%%%%%%%%%
\subsection{Preliminaries}
Before proving the theorems, we prepare two lemmas.
Let us introduce a subspace
\begin{align*}
    \mathcal{G}_L
    \coloneqq \bigcap_{j=1}^{L-2}\ker(1+\sigma_j^z)(1-\sigma_{j+1}^x\sigma_{j+2}^x).
\end{align*}
of the Hilbert space $(\mathbb{C}^2)^{\otimes L}$.
Note that $H\qty(\alpha=\hf)$ is written as the sum of $(1+\sigma_j^z)(1-\sigma_{j+1}^x\sigma_{j+2}^x)/4$ except for boundary terms or a constant:
\begin{align*}
    H\qty(\alpha=\hf)
    = \frac{1}{4}\sum_{j=1}^{L-2}(1+\sigma_j^z)(1-\sigma_{j+1}^x\sigma_{j+2}^x) + \cdots.
\end{align*}

\begin{lem}\label{lem:generic_jl}
    For any $\ket{\Omega}\in\mathcal{G}_L$, it holds that
    \begin{align}\label{eq:generic_jl}
        (1+\sigma_j^z)(1-\sigma_{\ell+1}^x\sigma_{\ell+2}^x)\ket{\Omega}
        = 0
    \end{align}
    for $1\le j \le \ell\le L-2$.
\end{lem}
\begin{proof}
    For fixed $j\in\{1,\ldots,L-2\}$, 
    we proceed by induction on $\ell\in\{j,\ldots,L-2\}$.
    By the definition of $\mathcal{G}_L$, \eqref{eq:generic_jl} holds for $\ell=j$.
    Now we suppose that \eqref{eq:generic_jl} holds for $\ell=m\in\{j,\ldots,L-3\}$, which reads
    \begin{align*}
        \sigma_{m+1}^x(1+\sigma_j^z)\ket{\Omega}
        = \sigma_{m+2}^x(1+\sigma_j^z)\ket{\Omega}.
    \end{align*}
    Combining this with
    \begin{align*}
        \sigma_{m+1}^z(1-\sigma_{m+2}^x\sigma_{m+3}^x)\ket{\Omega}
        = -(1-\sigma_{m+2}^x\sigma_{m+3}^x)\ket{\Omega},
    \end{align*}
    we have
    \begin{align*}
        0
        &= \{\sigma_{m+1}^z,\sigma_{m+1}^x\}(1+\sigma_j^z)(1-\sigma_{m+2}^x\sigma_{m+3}^x)\ket{\Omega} \notag\\
        &= -2\sigma_{m+2}^x(1+\sigma_j^z)(1-\sigma_{m+2}^x\sigma_{m+3}^x)\ket{\Omega},
    \end{align*}
    which follows \eqref{eq:generic_jl} for $\ell=m+1$.
\end{proof}

To emphasize the Hilbert space containing them, we introduce
\begin{align*}
    \ket{j,\rightarrow}_L\coloneqq\ket{\downarrow}^{\otimes (j-1)}\otimes\ket{\uparrow}\otimes\ket{\rightarrow}^{\otimes (L-j)} \in (\mathbb{C}^2)^{\otimes L}
\end{align*}
and
\begin{align*}
    \ket{j,\leftarrow}_L\coloneqq\ket{\downarrow}^{\otimes (j-1)}\otimes\ket{\uparrow}\otimes\ket{\leftarrow}^{\otimes (L-j)} \in (\mathbb{C}^2)^{\otimes L}
\end{align*}
for $1\le j\le L-2$.
%%%%%%%%%%%%%%%%%%%%%%%%%%%%%%%%%%%%%%%%%%%%%%%%%%%%%%%%%%%
\begin{lem}\label{lem:all_elements_of_G_L}
    \begin{align}\label{eq:all_elements_of_G_L}
        \mathcal{G}_L
        = \sum_{j=1}^{L-2}\left(\mathbb{C}\ket{j,\rightarrow}_L+\mathbb{C}\ket{j,\leftarrow}_L\right)
        + \ket{\downarrow}^{\otimes (L-2)}\otimes(\mathbb{C}^2)^{\otimes 2}.
    \end{align}
\end{lem}
\begin{proof}
    We can check $(\text{LHS})\supset(\text{RHS})$ by calculating
    \begin{gather*}
        (1+\sigma_\ell^z)(1-\sigma_{\ell+1}^x\sigma_{\ell+2}^x)\ket{j,\rightarrow}_L
        = 0, \\
        (1+\sigma_\ell^z)(1-\sigma_{\ell+1}^x\sigma_{\ell+2}^x)\ket{j,\leftarrow}_L
        = 0, \\
        (1+\sigma_\ell^z)(1-\sigma_{\ell+1}^x\sigma_{\ell+2}^x)\left(\ket{\downarrow}^{\otimes(L-2)}\otimes\ket{\psi}\right)
        = 0
    \end{gather*}
    directly for any $1\le\ell\le L-2$, $1\le j \le L-2$, and $\ket{\psi}\in(\mathbb{C}^2)^{\otimes 2}$.

    The proof of $(\text{LHS})\subset(\text{RHS})$ is by induction on $L\ge 3$. 
    We observe 
    \begin{align*}
        \mathcal{G}_3
        &\coloneqq \ker(1+\sigma_1^z)(1-\sigma_2^x\sigma_3^x) \\
        &= \mathbb{C}\ket{1,\rightarrow}_3 + \mathbb{C}\ket{1,\leftarrow}_3 + \ket{\downarrow}\otimes(\mathbb{C}^2)^{\otimes 2}
    \end{align*}
    by straightforward calculation.
    Now we suppose that 
    \begin{align}\label{eq:inductive_assumption}
        \mathcal{G}_L
        \subset \sum_{j=1}^{L-2}\left(\mathbb{C}\ket{j,\rightarrow}_L+\mathbb{C}\ket{j,\leftarrow}_L\right)
        + \ket{\downarrow}^{\otimes (L-2)}\otimes(\mathbb{C}^2)^{\otimes 2}
    \end{align}
    holds for an integer $L\ge 3$.
    We pick a vector $\ket{\Omega}_{L+1}\in\mathcal{G}_{L+1}$.
    From Lemma \ref{lem:generic_jl}, it holds that
    \begin{align*}
        (1-\sigma_{\ell+1}^x\sigma_{\ell+2}^x)(1+\sigma_1^z)\ket{\Omega}_{L+1}
        = 0
    \end{align*}
    for any $1\le \ell\le L-2$, which leads to
    \begin{align*}
        (1+\sigma_1^z)\ket{\Omega}_{L+1}
        \in \mathbb{C}^2\otimes\ket{\rightarrow}^{\otimes L} + \mathbb{C}^2\otimes\ket{\leftarrow}^{\otimes L}.
    \end{align*}
    Since the image and kernel of $1+\sigma_1^z$ in $(\mathbb{C}^2)^{\otimes(L+1)}$ are
    \begin{align*}
        \operatorname{im}(1+\sigma_1^z)
        = \ket{\uparrow}\otimes(\mathbb{C}^2)^{\otimes L},\ 
        \ker(1+\sigma_1^z)
        = \ket{\downarrow}\otimes(\mathbb{C}^2)^{\otimes L},
    \end{align*}
    we have 
    \begin{align*}
        \ket{\Omega}_{L+1}
        &\in \mathbb{C}\ket{\uparrow}\otimes\ket{\rightarrow}^{\otimes L} + \mathbb{C}\ket{\uparrow}\otimes\ket{\leftarrow}^{\otimes L} + \ket{\downarrow}\otimes(\mathbb{C}^2)^{\otimes L} \notag\\
        &= \mathbb{C}\ket{1,\rightarrow}_{L+1} + \mathbb{C}\ket{1,\leftarrow}_{L+1} + \ket{\downarrow}\otimes(\mathbb{C}^2)^{\otimes L};
    \end{align*}
    that is, there exist scalars $a,b\in\mathbb{C}$, and a vector $\ket{\Omega}_L\in(\mathbb{C}^2)^{\otimes L}$ such that
    \begin{align*}
        \ket{\Omega}_{L+1}
        = a\ket{1,\rightarrow}_{L+1} + b\ket{1,\leftarrow}_{L+1} + \ket{\downarrow}\otimes\ket{\Omega}_L.
    \end{align*}
    %%%%%%%%%%%%%%%%%%%%%%%%%%%%%%%%%%%%%%%%%%%%%%%%%%%%%%%
    Since 
    \begin{align*}
        0
        &= (1+\sigma_{j+1}^z)(1-\sigma_{j+2}^x\sigma_{j+3}^x)\ket{\Omega}_{L+1} \notag\\
        &= \ket{\downarrow}\otimes(1+\sigma_j^z)(1-\sigma_{j+1}^x\sigma_{j+2}^x)\ket{\Omega}_L
    \end{align*}
    for $1\le j\le L-2$, $\ket{\Omega}_L$ is a vector belonging to the subspace $\mathcal{G}_L$.
    The inductive assumption \eqref{eq:inductive_assumption} results in
    \begin{align*}
        \ket{\Omega}_L \in \sum_{j=1}^{L-2}\left(\mathbb{C}\ket{j,\rightarrow}_L+\mathbb{C}\ket{j,\leftarrow}_L\right)
        + \ket{\downarrow}^{\otimes (L-2)}\otimes(\mathbb{C}^2)^{\otimes 2}
    \end{align*}
    and thus
    \begin{widetext}
        \begin{align*}
        \ket{\Omega}_{L+1} &\in \mathbb{C}\ket{1,\rightarrow}_{L+1} + \mathbb{C}\ket{1,\leftarrow}_{L+1} +\ket{\downarrow}\otimes\Bigg(\sum_{j=1}^{L-2}\left(\mathbb{C}\ket{j,\rightarrow}_L+\mathbb{C}\ket{j,\leftarrow}_L\right) + \ket{\downarrow}^{\otimes (L-2)}\otimes(\mathbb{C}^2)^{\otimes 2}\Bigg) \nt
        &= \sum_{j=1}^{L-1}\left(\mathbb{C}\ket{j,\rightarrow}_{L+1}+\mathbb{C}\ket{j,\leftarrow}_{L+1}\right) + \ket{\downarrow}^{\otimes(L-1)}\otimes(\mathbb{C}^2)^{\otimes 2},
    \end{align*}
    which indicates that
    \begin{align*}
        \mathcal{G}_{L+1}
        \subset \sum_{j=1}^{L-1}\left(\mathbb{C}\ket{j,\rightarrow}_{L+1}+\mathbb{C}\ket{j,\leftarrow}_{L+1}\right)
        + \ket{\downarrow}^{\otimes (L-1)}\otimes(\mathbb{C}^2)^{\otimes 2}.
    \end{align*}
    \end{widetext}
\end{proof}

Finding the relations
\begin{gather*}
    \ket{\rightarrow}^{\otimes L}
    = \sum_{j=1}^{L-2}\left(\frac{1}{\sqrt{2}}\right)^j \ket{j,\rightarrow}_L + \left(\frac{1}{\sqrt{2}}\right)^{L-2}\ket{\downarrow\cdots\downarrow\rightarrow\rightarrow}, \\
    \ket{\leftarrow}^{\otimes L}
    = \sum_{j=1}^{L-2}\left(\frac{1}{\sqrt{2}}\right)^j \ket{j,\leftarrow}_L + \left(\frac{1}{\sqrt{2}}\right)^{L-2}\ket{\downarrow\cdots\downarrow\leftarrow\leftarrow},
\end{gather*}
we obtain another expression
\begin{align*}
    \mathcal{G}_L
    = &\mathbb{C}\ket{\rightarrow}^{\otimes L} + \mathbb{C}\ket{\leftarrow}^{\otimes L}  \notag\\ 
    &+ \sum_{j=2}^{L-2}\left(\mathbb{C}\ket{j,\rightarrow}_L+\mathbb{C}\ket{j,\leftarrow}_L\right) + \ket{\downarrow}^{\otimes (L-2)}\otimes(\mathbb{C}^2)^{\otimes 2}
\end{align*}
of Lemma \ref{lem:all_elements_of_G_L}.
Note that 
\begin{align}\label{eq:another_form_of_G_L}
    \mathcal{G}_L
    \subset \mathbb{C}\ket{\rightarrow}^{\otimes L} + \mathbb{C}\ket{\leftarrow}^{\otimes L} + \ket{\downarrow}\otimes(\mathbb{C}^2)^{\otimes (L-1)}.
\end{align}
%%%%%%%%%%%%%%%%%%%%%%%%%%%%%%%%%%%%%%%%%%%%%%%%%%%%%%%%%%%
\subsection{Proof of Theorem \ref{thm:PBC}}
The Hamiltonian $H\qty(\alpha=\hf)$ under the periodic boundary condition is of the form
\begin{align*}
    H\qty(\alpha=\hf)
    &= \frac{1}{4}\sum_{j=1}^{L-2}(1+\sigma_j^z)(1-\sigma_{j+1}^x\sigma_{j+2}^x) \notag\\
    &\quad+ \frac{1}{4}(1+\sigma_{L-1}^z)(1+(-1)^F\sigma_L^x\sigma_1^x) \notag\\
    &\quad+ \frac{1}{4}(1+\sigma_L^z)(1-\sigma_1^x\sigma_2^x)
\end{align*}
and every term is positive-semidefinite:
\begin{gather*}
    (1+\sigma_j^z)(1-\sigma_{j+1}^x\sigma_{j+2}^x)\ge 0 \quad\text{for}\ 1\le j\le L-2, \\
    (1+\sigma_{L-1}^z)(1+(-1)^F\sigma_L^x\sigma_1^x)\ge 0,
    \quad
    (1+\sigma_L^z)(1-\sigma_1^x\sigma_2^x)\ge 0,
\end{gather*}
which indicates that 
\begin{align*}
    \mathcal{E}_{\mathrm{PBC}}
    = \mathcal{G}_L &\cap \ker(1+\sigma_{L-1}^z)(1+(-1)^F\sigma_L^x\sigma_1^x) \notag\\
    &\cap \ker(1+\sigma_L^z)(1-\sigma_1^x\sigma_2^x).
\end{align*}

\begin{lem}\label{lem:Z_L-1X_1X_2}
    Any $\ket{\Omega}\in \ker(1+\sigma_{L-1}^z)(1+(-1)^F\sigma_L^x\sigma_1^x) \cap \ker(1+\sigma_L^z)(1-\sigma_1^x\sigma_2^x)$ satisfies
    \begin{align}\label{eq:Z_L-1X_1X_2}
        (1+\sigma_{L-1}^z)(1-\sigma_1^x\sigma_2^x)\ket{\Omega}
        = 0.
    \end{align}
\end{lem}
\begin{proof}
    This proof is similar to the proof of Lemma \ref{lem:generic_jl}.
    Since $\ket{\Omega}$ satisfies
    \begin{gather*}
        \sigma_L^z(1-\sigma_1^x\sigma_2^x)\ket{\Omega}
        = -(1-\sigma_1^x\sigma_2^x)\ket{\Omega}, \\
        \sigma_L^x(1+\sigma_{L-1}^z)\ket{\Omega}
        = (-1)^F \sigma_1^x(1+\sigma_{L-1}^z)\ket{\Omega},
    \end{gather*}
    we have
    \begin{align*}
        0
        &= \{\sigma_L^z,\sigma_L^x\}(1+\sigma_{L-1}^z)(1-\sigma_1^x\sigma_2^x)\ket{\Omega} \notag\\
        &= -2(-1)^F \sigma_1^x(1+\sigma_{L-1}^z)(1-\sigma_1^x\sigma_2^x)\ket{\Omega},
    \end{align*}
    which follows \eqref{eq:Z_L-1X_1X_2}.
\end{proof}

Based on the above preparation, let us prove Theorem \ref{thm:PBC}, in particular the expression \eqref{eq:thmPBC}.

We can check $\mathcal{E}_{\mathrm{PBC}}\supset\mathbb{C}\ket{\mathrm{A}}+\mathbb{C}\ket{\mathrm{B}}$ by calculating
\begin{align*}
    H\qty(\alpha=\hf)\ket{\mathrm{A}}
    = H\qty(\alpha=\hf)\ket{\mathrm{B}}
    = 0
\end{align*}
directly.

In order to prove $\mathcal{E}_{\mathrm{PBC}}\subset\mathbb{C}\ket{\mathrm{A}}+\mathbb{C}\ket{\mathrm{B}}$, we pick a vector $\ket{\Omega}\in\mathcal{E}_{\mathrm{PBC}}$.
Due to \eqref{eq:another_form_of_G_L}, $\ket{\Omega}\in\mathcal{G}_L$ is written as
\begin{align*}
    \ket{\Omega}
    = a\ket{\rightarrow}^{\otimes L} + b\ket{\leftarrow}^{\otimes L} + \ket{\downarrow}\otimes\ket{\Psi}_{L-1}
\end{align*}
with scalars $a,b\in\mathbb{C}$ and a vector $\ket{\Psi}_{L-1}\in(\mathbb{C}^2)^{\otimes(L-1)}$.
%%%%%%%%%%%%%%%%%%%%%%%%%%%%%%%%%%%%%%%%%%%%%%%%%%%%%%%%%%%
It also holds that
\begin{align*}
    0
    &= (1+\sigma_{L-1}^z)(1-\sigma_1^x\sigma_2^x)\ket{\Omega} \notag\\
    &= \ket{\downarrow}\otimes (1+\sigma_{L-2}^z)\ket{\Psi}_{L-1} - \ket{\uparrow}\otimes(1+\sigma_{L-2}^z)\sigma_1^x\ket{\Psi}_{L-1}
\end{align*}
from Lemma \ref{lem:Z_L-1X_1X_2}, 
\begin{align*}
    0
    &= (1+\sigma_L^z)(1-\sigma_1^x\sigma_2^x)\ket{\Omega} \notag\\
    &= \ket{\downarrow}\otimes (1+\sigma_{L-1}^z)\ket{\Psi}_{L-1} - \ket{\uparrow}\otimes(1+\sigma_{L-1}^z)\sigma_1^x\ket{\Psi}_{L-1},
\end{align*}
and 
\begin{align*}
    0
    &= (1+\sigma_j^z)(1-\sigma_{j+1}^x\sigma_{j+2}^x)\ket{\Omega} \notag\\
    &= \ket{\downarrow}\otimes(1+\sigma_{j-1}^z)(1-\sigma_j^x\sigma_{j+1}^x)\ket{\Psi}_{L-1}
\end{align*}
for $2\le j\le L-2$,
which leads to the constraints
\begin{gather}
    (1+\sigma_{L-2}^z)\ket{\Psi}_{L-1}
    = (1+\sigma_{L-1}^z)\ket{\Psi}_{L-1}
    = 0, \label{eq:constraint1}\\
    (1+\sigma_j^z)(1-\sigma_{j+1}^x\sigma_{j+2}^x)\ket{\Psi}_{L-1}
    = 0
    \quad
    \text{for}\ 1\le j\le L-3. \label{eq:constraint2}
\end{gather}
The constraint \eqref{eq:constraint1} enforces that there exists $\ket{\Psi}_{L-3}\in(\mathbb{C}^2)^{L-3}$ such that
\begin{align}\label{eq:constraint3}
    \ket{\Psi}_{L-1}
    = \ket{\Psi}_{L-3}\otimes\ket{\downarrow\downarrow}.
\end{align}

\begin{lem}\label{lem:all_down}
    If a vector $\ket{\Phi}_L\in(\mathbb{C}^2)^L$ satisfies
    \begin{align}\label{eq:constraint4}
        (1+\sigma_j^z)(1-\sigma_{j+1}^x\sigma_{j+2}^x)(\ket{\Phi}_L\otimes\ket{\downarrow\downarrow})
        = 0
    \end{align}
    for $1\le j\le L$,
    then $\ket{\Phi}_L\in\mathbb{C}\ket{\downarrow}^{\otimes L}$.
\end{lem}
\begin{proof}
    This proof is by induction on $L\ge 1$.
    For $L=1$, we have
    \begin{align*}
        0
        &= (1+\sigma_1^z)(1-\sigma_2^x\sigma_3^x)(\ket{\Phi}_1\otimes\ket{\downarrow\downarrow}) \notag\\
        &= (1+\sigma_1^z)\ket{\Phi}_1\otimes\ket{\downarrow\downarrow} - (1+\sigma_1^z)\ket{\Phi}_1\otimes\ket{\uparrow\uparrow},
    \end{align*}
    which results in $(1+\sigma_1^z)\ket{\Phi}_1 = 0$, i.e., $\ket{\Phi}_1\in\mathbb{C}\ket{\downarrow}$.

    For an integer $L=N+1\ge 2$, we have
    \begin{align*}
        0
        &= (1+\sigma_L^z)(1-\sigma_{L+1}^x\sigma_{L+2}^x)(\ket{\Phi}_L\otimes\ket{\downarrow\downarrow}) \notag\\
        &= (1+\sigma_{N+1}^z)\ket{\Phi}_{N+1}\otimes\ket{\downarrow\downarrow} - (1+\sigma_{N+1}^z)\ket{\Phi}_{N+1}\otimes\ket{\uparrow\uparrow},
    \end{align*}
    which results in $(1+\sigma_{N+1}^z)\ket{\Phi}_{N+1} = 0 $, i.e., $\ket{\Phi}_{N+1} = \ket{\Phi}_N\otimes\ket{\downarrow}$ with a vector $\ket{\Phi}_N\in(\mathbb{C}^2)^N$. 
    Since 
    \begin{align*}
        0
        &= (1+\sigma_j^z)(1-\sigma_{j+1}^x\sigma_{j+2}^x)(\ket{\Psi}_{N+1}\otimes\ket{\downarrow\downarrow}) \notag\\
        &= [(1+\sigma_j^z)(1-\sigma_{j+1}^x\sigma_{j+2}^x)(\ket{\Phi}_N\otimes\ket{\downarrow\downarrow})]\otimes\ket{\downarrow}
    \end{align*}
    for $1\le j\le N$, 
    the vector $\ket{\Phi}_N$ satisfies \eqref{eq:constraint4}.
    
    Now we suppose that \eqref{eq:constraint4} holds for $L=N$.
    Then we obtain $\ket{\Phi}_N\in \mathbb{C}\ket{\downarrow}^{\otimes N}$, i.e., $\ket{\Phi}_{N+1} \in \mathbb{C}\ket{\downarrow}^{\otimes(N+1)}$. 
\end{proof}

From Lemma \ref{lem:all_down}, we have
\begin{align*}
    \ket{\Omega}
    = a\ket{\rightarrow}^{\otimes L} + b\ket{\leftarrow}^{\otimes L} + c\ket{\downarrow}^{\otimes L}.
\end{align*}
with a scalar $c\in\mathbb{C}$.
Since 
\begin{align*}
    0
    &= (1+\sigma_{L-1}^z)(1+(-1)^F\sigma_L^x\sigma_1^x)\ket{\Omega} \notag\\
    &= (a + (-1)^L b)(1+\sigma_{L-1}^z)\ket{\rightarrow}^{\otimes L} \notag\\
    &\quad + ((-1)^L a + b)(1+\sigma_{L-1}^z)\ket{\leftarrow}^{\otimes L},
\end{align*}
we obtain $b = -(-1)^L a$ and thus
\begin{align*}
    \ket{\Omega}
    = \sqrt{2}a\ket{\mathrm{A}} + c\ket{\mathrm{B}},
\end{align*}
which indicates that $\mathcal{E}_{\mathrm{PBC}}\subset\mathbb{C}\ket{\mathrm{A}}+\mathbb{C}\ket{\mathrm{B}}$.

Therefore Theorem \ref{thm:PBC} has been proven.
%%%%%%%%%%%%%%%%%%%%%%%%%%%%%%%%%%%%%%%%%%%%%%%%%%%%%%%%%%%
\subsection{Proof of Theorem \ref{thm:OBC}}
The Hamiltonian $H\qty(\alpha=\hf)$ under the open boundary condition is of the form
\begin{align*}
    H\qty(\alpha=\hf)
    = \frac{1}{4}\sum_{j=1}^{L-2}(1+\sigma_j^z)(1-\sigma_{j+1}^x\sigma_{j+2}^x) + \frac{1}{8}
\end{align*}
and every term is positive-semidefinite:
\begin{align*}
    (1+\sigma_j^z)(1-\sigma_{j+1}^x\sigma_{j+2}^x)\ge 0 \quad\text{for}\ 1\le j\le L-2,
\end{align*}
which indicates that $\mathcal{E}_{\mathrm{OBC}}=\mathcal{G}_L$.
This equality and Lemma \ref{lem:all_elements_of_G_L} result in the statement in Theorem \ref{thm:OBC}.

%%%%%% REFERENCES %%%%%%%%%%%%%%%%%%%%%%%%%%%%%%%%%%%%%
%\bibliographystyle{apsrev4-2}
%\bibliography{references.bib} 
%
%%%%%%%%%%%%%%%%%%%%%%%%%%%%%%%%%%%%%%%%%%%%%%%%%%%
\end{document}